\documentclass[10pt,journal]{IEEEtran}
\IEEEoverridecommandlockouts

\usepackage{epsfig}
\usepackage{amssymb} 
\usepackage[tbtags]{amsmath} 
\usepackage{graphics,eepic,epic}
\usepackage{latexsym}
\usepackage{euscript}
\usepackage{graphics,eepic,epic,psfrag}
\usepackage{textcomp}
\usepackage{array}
\usepackage{hhline}
\usepackage{bm}

\usepackage{float}
\newfloat{algorithm}{t}{lop}
\newfloat{algorithm*}{t}{lop}

\usepackage[defaultlines=4,all]{nowidow}

\usepackage{algorithm}
\usepackage{algpseudocode}
\algnewcommand\algorithmicinput{\textbf{Input:}}
\algnewcommand\algorithmicoutput{\textbf{Output:}}
\algnewcommand\algorithmicuses{\textbf{Uses:}}
\algnewcommand\Input{\item[\algorithmicinput]}%
\algnewcommand\Output{\item[\algorithmicoutput]}%
\algnewcommand\Uses{\item[\algorithmicuses]}%

\usepackage{framed}
\usepackage{mdframed}
\usepackage{afterpage}
\usepackage{adjustbox}

\usepackage[at]{easylist}

\newcommand{\SHORTTEX}{} 

\newcommand{\eqbreak}[1][2]{\\&\hskip#1em}

\usepackage[style=ieee,backend=bibtex,bibstyle=numeric-comp,sortcites=true,natbib=true,maxcitenames=20,maxnames=20]{biblatex}

\renewcommand{\cite}{\unspace~\autocite}

\bibliography{fast_computation_3d_separable_arxiv-v3_biber.bib}

\usepackage{mathtools}
\usepackage{esdiff}

\usepackage{amsthm}

\usepackage{appendix}

\makeatletter
\def\@fnsymbol#1{\ensuremath{\ifcase#1\or *\or \dagger\or \ddagger\or
   \mathsection\or \mathparagraph\or \|\or **\or \dagger\dagger
   \or \ddagger\ddagger \else\@ctrerr\fi}}
\makeatother

\usepackage{fancyhdr}

\usepackage{empheq}

\usepackage{mathtools}
\usepackage{esdiff}

\makeatletter
\def\@fnsymbol#1{\ensuremath{\ifcase#1\or *\or \dagger\or \ddagger\or
   \mathsection\or \mathparagraph\or \|\or **\or \dagger\dagger
   \or \ddagger\ddagger \else\@ctrerr\fi}}
\makeatother

\renewcommand{\d}[1]{\ensuremath{\operatorname{d}\!{#1}}}

\def\R{\mathbb{R}}

\DeclarePairedDelimiter{\abs}{\lvert}{\rvert}
\newcommand{\norm}[1]{\left\lVert#1\right\lVert}

\DeclareMathOperator{\rect}{rect}
\DeclareMathOperator{\step}{u}
\DeclareMathOperator{\signum}{sign}

\newcommand{\bfx}{\mathbf{x}}

\newcommand{\bfe}{\mathbf{e}}
\newcommand{\bfu}{\mathbf{u}}

\newcommand{\bfa}{\mathbf{a}}
\newcommand{\bfb}{\mathbf{b}}
\newcommand{\bfg}{\mathbf{g}}

\newcommand{\bfn}{\mathbf{n}}
\newcommand{\bfzero}{\mathbf{0}}

\newcommand{\bfomega}{\bm{\omega}}
\newcommand{\dx}{\d{x}}
\newcommand{\dy}{\d{y}}

\newcommand{\dt}{\d{t}}
\newcommand{\dr}{\d{r}}

\newcommand{\sfF}{\mathsf{F}}
\newcommand{\Fp}{\mathsf{F}_\text{P}}
\newcommand{\Fo}{\mathsf{F}_\text{0}}

\newcommand{\scrT}{\mathcal{T}}
\newcommand{\scrG}{\mathcal{G}}

\newcommand{\scrS}{\mathcal{S}}
\newcommand{\scrH}{\mathcal{H}}

\newcommand{\convast}{\circledast}
\newcommand{\conv}[1]{\underset{#1}{\convast}}

\def\PSI{\eta} 
\def\PHI{\Phi}

\newtheorem{theorem}{Theorem}[section]

\newtheorem{proposition}[theorem]{Proposition}

\usepackage[switch]{lineno}

\newcommand*\patchAmsMathEnvironmentForLineno[1]{%
  \expandafter\let\csname old#1\expandafter\endcsname\csname #1\endcsname
  \expandafter\let\csname oldend#1\expandafter\endcsname\csname end#1\endcsname
  \renewenvironment{#1}%
     {\linenomath\csname old#1\endcsname}%
     {\csname oldend#1\endcsname\endlinenomath}}%
\newcommand*\patchBothAmsMathEnvironmentsForLineno[1]{%
  \patchAmsMathEnvironmentForLineno{#1}%
  \patchAmsMathEnvironmentForLineno{#1*}}%
\AtBeginDocument{%
\patchBothAmsMathEnvironmentsForLineno{equation}%
\patchBothAmsMathEnvironmentsForLineno{align}%
\patchBothAmsMathEnvironmentsForLineno{flalign}%
\patchBothAmsMathEnvironmentsForLineno{alignat}%
\patchBothAmsMathEnvironmentsForLineno{gather}%
\patchBothAmsMathEnvironmentsForLineno{multline}%
}

\usepackage{tikz}
\newcommand\copyrighttext{%
\tiny 
  \textcopyright 2020 IEEE. Personal use of this material is permitted.
  Permission from IEEE must be obtained for all other uses, in any current or future
  media, including reprinting/republishing this material for advertising or promotional
  purposes, creating new collective works, for resale or redistribution to servers or
  lists, or reuse of any copyrighted component of this work in other works.    
  Please cite the published version of this work as: 
  Zalev, J. and Kolios, M. C., 
  "Fast 3-D Opto-Acoustic Simulation for Linear Array With Rectangular Elements."  
  IEEE Transactions on Ultrasonics, Ferroelectrics, and Frequency Control 68.5 (2020): 1885-1906.  
  DOI: 10.1109/TUFFC.2020.3038622%
  }

\newcommand\copyrightnotice{%
\begin{tikzpicture}[remember picture,overlay]
\node[anchor=south,yshift=10pt] at (current page.south) {\fbox{\parbox{\dimexpr\textwidth-\fboxsep-\fboxrule\relax}{\copyrighttext}}};
\end{tikzpicture}%
}

\usepackage{csquotes}

\usepackage{subfig}

\begin{document}

\label{sec:title}

\title
{%
Fast Three-dimensional Opto-acoustic Simulation \\for Linear Array with Rectangular Elements
}

\author{
 Jason Zalev
}

\author
{
 Michael C. Kolios
}

\author{
\IEEEauthorblockN{Jason Zalev and Michael C. Kolios}

\thanks{J. Zalev and M. C. Kolios are with the Department of Physics,  Ryerson University, Toronto, ON, Canada
(email: jzalev@ryerson.ca, mkolios@ryerson.ca)}

}

\date{Draft \today}

\maketitle

\copyrightnotice


\begin{abstract}
\label{sec:abstract}
\addcontentsline{toc}{section}{Abstract}
Simulation involves predicting 
responses of
 a physical system. 
In this article, 
we simulate
opto-acoustic signals
generated
in a three-dimensional volume
due to 
absorption of an optical pulse. 
A separable computational model is developed, which splits processing into two steps,  
permitting
 an order-of-magnitude 
improvement in computational efficiency 
over
a non-separable model. 
The simulated signals 
represent
acoustic waves,
measured
by a probe 
with a linear transducer array,
in a rotated and translated coordinate frame.
Light is delivered 
by an optical source 
that moves with
the probe's  frame. 
A spatio-temporal impulse response for rectangular-element transducer geometry is 
derived using a 
 Green's function solution to the acoustic wave equation. 
The approach permits fast and accurate simulation 
for a 
probe 
with arbitrary trajectory, 
which is useful for modeling free-hand acquisitions. 
For a 3D volume of $n^3$ voxels, computation is accelerated by a factor of $n$. 
This  may potentially  have application in opto-acoustic imaging, 
where clinicians 
 visualize 
structural and functional features
of biological tissue
for 
 assessment of cancer and other diseases.

\end{abstract}


\section{Introduction}
\label{sec:intro}
Acoustic waves are generated when light heats optically absorbing materials. 
In the case of
biological tissues, 
near-infrared
laser pulses will
 produce strong opto-acoustic signals from blood and other absorbing structures. 
To permit visualization of 
 cancer and other diseases, transducer arrays must detect 
 these
  signals to form images. 
Rapid simulation 
of this process
 may be required 
during 
image formation
 to solve 
 an inverse problem,  
 or for performing standalone analysis. 
 
Accurate simulation of
the system's physics, 
including 
modeling of 
transducer elements and their 
impulse responses, is
critical to avoid unwanted distortions that may arise in practice. 
 However, the processing time required for 3D simulation
can be a limiting factor, especially when a more accurate transducer model is used.  
  Here we provide a mathematical derivation and approach for performing fast opto-acoustic simulation of a 3D volume 
for a linear array with rectangular elements. 

In our approach, a forward (and/or adjoint) operator is computed in the time-domain using a 
\emph{separable cascade}. We prove that a non-separable 3D Green's function model of the system response is equivalent to a cascaded series of two perpendicular 2D operations, even when the optical energy distribution, and spatial impulse response of the transducer's rectangular aperture are included. 
This can greatly improve the computational efficiency. 
For a linear array, 
time-domain data  is computed efficiently due to shift-invariance in the cascaded operation.
 Furthermore, the linear array can have arbitrary position and orientation, relative to the 3D volume of opto-acoustic sources. This facilitates simulating multiple frames of data when the linear array moves along an arbitrary trajectory 
 to record data for several laser pulses. 
In clinical imaging, this type of motion occurs during free-hand scanning with a probe at the surface of a subject's skin. 
Accordingly, we plan to 
incorporate 
the proposed technique 
towards 
improving image quality 
for
opto-acoustic image reconstruction. 

\subsection{Outline}
In Section~\ref{sec:intro}, 
we introduce 
the context of fast opto-acoustic simulation relative to other approaches in the literature, 
and give an overview of our approach. 

In Section~\ref{sec:model}, a mathematical model that defines opto-acoustic simulation for a linear array probe is described. 
The model is represented as a \emph{mathematical operator} that includes the effects of: i) acoustic wave propagation for a 3D volume with constant speed of sound and density; ii) spatial impulse response due to finite transducer aperture on the received signal (e.g. effects of rectangular transducer elements and directionality); and iii) a spatially dependent optical fluence distribution of the delivered light energy. We show how the probe's position and orientation, in a coordinate frame relative to the volume's coordinates, influences the resulting signal. We also demonstrate mathematical properties that are necessary for the following sections. 

In Section~\ref{sec:separable}, we prove that the mathematical operator developed in Section~\ref{sec:model} is separable. 
An exact  analytic solution, which is separable, is formulated for when rectangular transducer elements are used. 
Furthermore, if an exact solution is not required, we also present a far-field approximation that 
permits additional performance increase. 

In Section~\ref{sec:impl}, we describe the results of our implementation for fast 3D simulation, and 
compare the simulated acoustic output to what is obtained with other 3D simulations. 
In Section~\ref{sec:discussion}, we discuss the model and analyze its computational complexity. We conclude in Section~\ref{sec:concl}.
A list of symbols and equations is provided in Table~\ref{tab:symbols}. 

\subsection{Significance} 
To our best knowledge, our work is the first that involves combining a compositionally-separable time-domain operator for acoustic simulation with an efficient method to include the spatial impulse response from multiplicatively-separable rectangular transducer elements. We also believe that this is the first work to describe a separable operator for 3D acoustic simulation in the rotated frame of a linear array probe, and to recognize that there is computational advantage of doing so during implementation. Also, our proof in Section~\ref{sec:separable}, which is based on showing separability using the Dirac delta function, is novel. It illustrates 
 the separability of the 3D acoustic wave Green's function, and generalizes to a wider class of Green's function kernels that are separable. Additionally, our opto-acoustic simulation model in Section~\ref{sec:model} is novel in the sense that the transducer aperture function and the optical energy distribution move in the coordinate reference frame of the linear array relative to the volume's coordinates, which is applicable for simulating transducer arrays undergoing motion.

\subsection{Relation to other work}

Opto-acoustic imaging is a new modality that has recently been investigated clinically  for improved accuracy in breast cancer diagnosis\cite{oraevsky2018clinical, zalev2015opto, neuschler2018pivotal, menezes2018downgrading, oraevsky2014optoacoustic, xu2006photoacoustic, zalev2018image}.  
The modality is able to visualize breast tumors several centimetres deep within tissue, 
and may potentially  improve distinction between benign and malignant lesions\cite{oraevsky2018clinical, zalev2015opto, neuschler2018pivotal, menezes2018downgrading}. The simulation model described in this article was developed by considering the geometry and specifications of a hand-held opto-acoustic probe with a linear transducer array and light delivery unit that would be suitable for clinical imaging.

There are numerous approaches for performing acoustic and opto-acoustic simulations. 
The approach proposed in our paper is formulated for opto-acoustic simulation in the \emph{time-domain}. 
For an omni-directional point detector, it is known that opto-acoustic wave propagation can be computed as a convolutional operator in the time-domain, based on a Green's function solution to the scalar wave equation
\cite{morse1968theoretical,oraevsky2014optoacoustic, xu2006photoacoustic}.
Several time-domain based approaches for acoustic simulation that involve rectangular elements\cite{ lasota1982application,selfridge1980theory,
jensen1996field,jensen1992calculation,
stepanishen1971transient, lasota1984acoustic, scarano1985new, san1992diffraction}
are briefly summarized in this section.

Separable acoustic wave methods have also been studied.  
The approach of 
Jakubowicz~et~al.~(1989)~\cite{jakubowicz1989two} showed that 
seismological 
migration for 3D acoustic wave fields  was separable in the time-domain using cascaded operations. 
The 
method 
involved
 a \emph{backward} (inversion)  operator 
  to  \mbox{\emph{reconstruct}} an imaging volume from its measured signals\cite{claerbout2008basic}. 
However, to \emph{simulate} signals from an imaging volume, a \emph{forward} operator is required. 
By implementing the equations of 
Jakubowicz~et~al.~\cite{jakubowicz1989two}
 in the opposite order, 
 we develop a 
 forward operator that is separable. 
  Furthermore, we show that efficient computation and separability is still possible when rectangular elements are used, and that the equation can be applied in a rotated coordinate system of the probe.

For approaches of performing simulations in the time-domain, the impulse response for a rectangular element has been well studied\cite{stepanishen1971transient, lasota1984acoustic, scarano1985new, san1992diffraction,jensen1992calculation, hunt1983ultrasound,
cobbold2006foundations}. 
 The exact time-domain impulse response for a rectangular element was described by 
San Emeterio and Ullate (1992)~\cite{san1992diffraction}. 
For a line source, the exact time-domain impulse response is described by Lasota et al. (1984)~\cite{lasota1984acoustic}.  
 Scarano et al. (1985)~\cite{scarano1985new} proposed 
an impulse response calculation for rectangular transducer elements 
that involved a separable aperture function. However, this 
involved
 multiplicative separability for rectangular geometry, 
 which is different from our approach that uses 
  separability of 
convolutional operators by composition.  
A trapezoidal far-field approximation for the impulse response of rectangular transducers was described by Stepanishen~(1971)~\cite{stepanishen1971transient,stepanishen1972comments} and Freedman~(1971)~\cite{freedman1971farfield}.
McGough~(2004)~\cite{mcgough2004rapid} describes rapid calculation of time-harmonic near-field pressures from rectangular elements, which improves numerical performance by removing singularities that arise in integral equations.

 Software for acoustic simulation is also available. 
 Field II is a program designed for computing pressure fields from arbitrarily shaped, apodized, and excited ultrasound transducers\cite{jensen1996field,jensen1992calculation}. 
 The approach subdivides an aperture into multiple rectangular elements. 
FOCUS is another toolbox, which uses an approach called the fast near-field method to efficiently compute accurate acoustic fields 
 for various geometries\cite{mcgough2004rapid,chen20082d,mcgough2004efficient}. 
In Field II, simulation of opto-acoustic signals from a 3D volume can be performed by forming a superposition of weighted impulse responses from a list of absorbers located at different 3D positions. This is less efficient than the proposed approach, which uses volumetric separability to compute output for a linear-array.

In addition to time-domain approaches, separable approaches for acoustic waves models have also been described for the 
\emph{frequency-domain}\cite{jakubowicz1983simple, fokkema1986comments}. 
This involves using a pair of two-dimensional cascaded operators (analogous to the cascaded time-domain operators). 
In each cascaded 2D frequency-domain operation, the space of 2D Fourier data is warped using a 
non-linear transformation.
The frequency-domain computations work best when elements are aligned along a 2D grid, which  is preferred for efficient implementation of the multi-dimensional fast Fourier transform operator.

Other frequency-domain methods for computing acoustic signals include pseudo-spectral or \enquote{k-space} approaches\cite{bojarski1982k,cox2007k,cox2005fast}. 
One such simulator, called K-wave\cite{treeby2010k},
is designed to produce opto-acoustic signals from a dense volume with non-homogeneous acoustic properties. 
For non-homogeneous volumes, k-space methods use temporally discretized steps 
that compute
repeated frequency-domain operations. 
Each step can be implemented using separable operations; however,   
general purpose k-space simulation of non-homogeneous volumes
 cannot exploit the same efficiency that arises in a separable model for the homogeneous case.

For opto-acoustic imaging, 
the effects of transducer aperture  
on image reconstruction
have previously been studied\cite{anastasio2007application,xu2003analytic}. 
In addition, models involving far-field approximations for rectangular elements 
using a frequency-domain approach
 have been investigated
  for use in opto-acoustic reconstruction\cite{wang2011imaging, mitsuhashi2014investigation, wang2013accelerating}. 
However, these methods do not involve a separable simulation model 
and 
are different from our proposed approach.

\section{Opto-acoustic Simulation Model}
\label{sec:model}
In this section, we describe our model for opto-acoustic simulation. 
Subsequently, in Section~\ref{sec:separable}, we will prove that it has an equivalent separable implementation that can be computed efficiently.
The model simulates the system-response of a 3D volume of optical absorbers, when  illuminated by a pulse of light. 
It is assumed that the source of light moves in the same coordinate reference frame as a transducer array that records opto-acoustic signals. 

In opto-acoustics, 
signal strength
is proportional to a volume's optical absorption coefficient and to its Gr{\"u}neisen coefficient. 
These 
 vary spatially within the 3D volume. 
When multiplied together, they form a parameter called the \emph{opto-acoustic conversion efficiency}, 
which is 
 treated as an input to our simulation model.  
The 
signal strength
is also proportional to the \emph{optical energy distribution}, 
another system input,
which describes the amount of light that reaches each location in the volume. 

In the model,  it is assumed 
that 
the optical energy distribution depends only on \emph{bulk} optical properties of the volume, 
independent from the volume's spatially dependent optical absorption coefficient. 
This 
 approximation  makes it possible
 to rotate and translate the spatial distribution of optical energy, 
which moves with the position and orientation of the light source.
It is also assumed that the acoustic properties of the volume are homogeneous (constant speed of sound, density, etc.).

The output from the simulation includes a set of time-domain acoustic signals corresponding to what would be measured by an array of rectangular transducer elements, positioned in the coordinate frame of the probe relative to the volume.

\ifdefined\SHORTTEX
\begin{table}[t!]
\caption{List of symbols and equations}
\begin{adjustbox}{height=0.48\textheight}
\centering
\begin{tabular}{rp{2.4in}}
\hline\hline\\
\else
\vfill
\subsection{List of symbols and equations}
\hfill
\begin{center}
\begin{tabular}{rl}
\fi

$\psi(\bfx)$
&  acoustic source distribution (initial excess pressure) \\
&$\qquad\psi(\bfx)=\PSI(\bfx)\PHI_{A,\bfb}(\bfx)$\\

$\PSI(\bfx)$
& opto-acoustic conversion efficiency \\

$\PHI(\bfx)$& optical energy distribution \\
&$\,\,\PHI_{A,\bfb}(\bfx)=\scrT_{A,\bfb}\{\PHI\}(\bfx)=\PHI\big(A^T(\bfx-\bfb)\big)$\\
\\

$p(\bfx,t)$& acoustic pressure \\
&$\qquad p(\bfx, t) = \scrH\{\psi\}(\bfx,t)$\\
$s(\bfu,t)$& simulated opto-acoustic signal \\
&$\qquad s(\bfu, t) = \scrH^{A,\bfb}_{f, \PHI}\{\PSI\}(\bfu,t)$\\

\\

$\scrH\{\psi\}(\bfx, t)$ & 
ideal system response operator\\
$\scrH_f\{\psi\}(\bfx, t)$ & 
system response operator for aperture $f$\\
$\scrH^{A,\bfb}\{\psi\}(\bfu, t)$ & 
system response operator in probe frame\\
&$\,\, \scrH_f^{A,\bfb}\{\psi\}(\bfu, t)=h_f(\bfu, t) \conv{\bfu} \psi(A\bfu+\bfb)$\\
$\scrH_{f,\PHI}^{A,\bfb}\{\PSI\}(\bfu, t)$ &
generalized system response operator\\
&$ \scrH_{f,\PHI}^{A,\bfb}\{\PSI\}(\bfu, t)=h_f(\bfu, t) \conv{\bfu} 
\big[\underbrace{\PSI(A\bfu+\bfb)\PHI(\bfu)}_{\psi^{A,\bfb}(\bfu)}\big]$\\

$\scrG\{\psi\}(\bfx, t)$ & 
factor of system response operator \\
&$\qquad \scrH^{A,\bfb}= \tilde \scrG \circ \scrG^{A,\bfb} $\\

\\

$f_\bfa(\bfx)$
&  rectangular aperture function \\
&$\qquad f_\bfa(\bfx)=\rect(\frac{x_1}{a_1},\frac{x_2}{a_2})\delta(x_3)$\\

$g(\bfx,t)$& 
d'Alembert free-space Green's function  \\
&$\qquad g(\bfx,t)= \frac{1}{4\pi t}\delta(c t-\norm\bfx) $\\

$h_f(\bfx,t)$
& pressure impulse response for aperture $f$\\
&$\quad h_f(\bfx, t)=\bigg[\underbrace{\frac{\partial}{\partial t}g(\bfx, t)}_{h(\bfx,t)}\bigg] \conv{\bfx} f(\bfx)$\\

$\bar h_f(\bfx, t)$
& modified pressure impulse response for aperture $f$\\
&$\quad \bar h_f(\bfx, t)=\big[
\underbrace{ h(\bfx, t)\alpha(\bfx)}_{\bar h(\bfx,t)} \big] \conv{\bfx} f(\bfx)$\\

$\alpha(\bfx)$
& obliquity function \\

\\
$\delta(t)$& Dirac impulse function \\
$\step(t)$& Heaviside step function \\
\\
\ifdefined\SHORTTEX
\else
\end{tabular}

\begin{tabular}{rl}
\fi

$\bfx$& global volume coordinates in $\sfF_0$ \\
&$\qquad \bfx=(x_1,x_2,x_3)
$\\
$\bfu$& local coordinates in probe frame $\sfF_\text{P}$ \\
&$\qquad\bfu=
(u_1,u_2,u_3)
$\\

\\
$\hat \bfe_1,\hat \bfe_2, \hat \bfe_3$ & standard basis vectors \\

$\hat \bfn_1,\hat \bfn_2, \hat \bfn_3$ & basis vectors for $\sfF_\text P$\\

\\
$\sfF_\text{0}$ & base coordinate reference frame \\
$\sfF_\text{P}$ & probe coordinate reference frame \\
\\
$T_{A,\bfb}(\bfu)$& affine transform function  \\
&$\qquad T_{A,\bfb}(\bfu) = A\bfu + \bfb$\\
$T^{-1}_{A,\bfb}(\bfx)$& inverse affine transform function  \\
&$\qquad T^{-1}_{A,\bfb}(\bfx) = A^T(\bfx - \bfb)$\\

$\scrT_{A,\bfb}\{\psi\}(\bfx)$
 & affine transform operator \\
&$\,\,
\scrT_{A,\bfb}^{-1}\{\psi\}(\bfu)=\psi\!\left( T_{A,\bfb} (\bfu)\right)
=\psi^{A,\bfb}(\bfu)
$\\

\\
$A$& rotation matrix  \\
$\bfb$& translation vector  \\
\\

$f\circledast \tilde f$& convolution\\
&$\quad [f\circledast \tilde f](\bfx)=\int_{\R^3} \!f(\bfx')\tilde f(\bfx-\bfx')\,\d\bfx' $\\
$f\circ \bfg\hspace{1.2pt}$& composition\\
&$\qquad  [f\circ\bfg](\bfx)=f(\bfg(\bfx))$\\
&$\quad  [\scrS\circ\tilde\scrS]\{\psi\}(\bfx)= \scrS\{ \tilde\scrS\{\psi\} \}(\bfx)$\\
$f\odot \tilde f$& pointwise multiplication \\
&$\qquad [f\odot \tilde f](\bfx)=f(\bfx)\tilde f(\bfx)$\\

\ifdefined\SHORTTEX
\\\hline\hline
\end{tabular}
\end{adjustbox}
\vspace{-10px}
\label{tab:symbols}
\end{table}
\else
\end{tabular}%
\label{tab:symbols}
\end{center}
\fi

\subsection{Notation and preliminaries}
\label{sec:notation}

In our notation, 
vectors and vector functions appear in a bold font 
(e.g. $\bfx,\, \bfg(\bfx), \ldots$). 
Scalars and scalar functions appear in plain font (e.g. $f,\,g(\bfx),\,k,\ldots\,$).
Spatial position is represented by a vector in Euclidean space. 
The components of a three-dimensional position vector $\bfx \in \R^3$ are written $x_1$, $x_2$, $x_3$. 
A matrix (e.g. $A$) is described with uppercase roman letters. 
The transpose of matrix $A$ is written $A^T$. 
Column vectors are written with round brackets, and row vectors use square brackets. 
Thus, 
$\bfu=(u_1,u_2,u_3)=[u_1,u_2,u_3]^T$ and $\bfu^T=[u_1,u_2,u_3]$. 
A function is often written without its argument.
For example, $\psi(\bfx)$ can be shortened to $\psi$, when usage is clear from the context. 
A mathematical operator (e.g. $\scrG, \scrH, \scrT, \ldots$), which maps one space of functions to another space of functions, is written with script.
When a function $\psi$ is transformed (into a new function) by an operator $\scrS$, this is written $\scrS\{\psi\}$. 
 The affine transform \emph{function} $T_{A,\bfb}$, which acts on a vector and produces a vector, is written in uppercase; it should not be confused with the affine transform \emph{operator} $\scrT_{A,\bfb}$ which acts on a function and produces a function, and is written in script. Also, since the number of letters in the alphabet is limited, a tilde symbol over a letter (e.g. $\tilde a, \tilde \scrG, \ldots$) is used to indicate a different (and possibly related) variable or operator.

\hfill

To perform acoustic simulation, 
operators involving \emph{convolution} and \emph{composition} are used. These are described below. 

\subsubsection{Composition}

 The composition of a function $f(\bfx)$ with a function ${\mathbf{g}}(\bfx)$ 
 is defined as 
\begin{align}
[f \circ {\bfg}](\bfx) \equiv f({\bfg}(\bfx)).
\end{align}

For composition to be valid, the domain of $f$ must match the result produced by $\bfg$. For example, if $f\colon\R^3 \rightarrow \R$ (which means the input  of $f$ is a vector in $\R^3$ and its output is a real number), then $\bfg$ must output a vector in $\R^3$ to match the domain of $f$.  

Composition can also be applied when the function is a matrix. For example, if $A\in \R^{3\times3}$, then $[f \circ A](\bfx) = f(A\bfx)$. 
Composition with the affine transform function is written \mbox{$[f \circ T_{A,\bfb}](\bfx)$}, which is equal to $f(T_{A,\bfb}(\bfx))$.

\hfill 

\subsubsection{Mathematical operators}

In this paper, composition is  applied to mathematical operators.   
We refer to the composition of two operators as a \emph{cascade}. For fast acoustic simulation, we are concerned with the cascade 
\mbox{$\scrH\{\psi\}(\bfx) = [\tilde\scrG  \circ \scrG]\{\psi\}(\bfx)$}. 
That is to say, $\scrH$ is \emph{compositionally-separable} into $\tilde \scrG$ and  $\scrG$, when applied to any source function $\psi$. 
Here, the composition of two  operators $\tilde\scrG$ and  $\scrG$ is defined  
\mbox{%
$[\tilde\scrG \circ \scrG]\{\psi\}(\bfx) \equiv \tilde\scrG  \{ \scrG \{\psi\}\}(\bfx)$}.

To be somewhat more precise,  
a  {transform operator} is a function that transforms 
one function-space $S_1$ into another function-space $S_2$. 
One familiar  transform operator is the (temporal) Fourier transform  $\mathcal F \colon (\R\!\rightarrow\!\R) \longrightarrow (\R\! \rightarrow \!\mathbb C)$,
which
 converts a real valued function in the time-domain $f(t)$  to a complex-valued function in the frequency domain $\hat f({\bfomega})$,
  with the relationship $\hat f({\bfomega}) = \mathcal F\{f \}({\bfomega})$.

When an operator $\mathcal{S} \colon S_1 \rightarrow S_2$ is applied to a function \mbox{$f\in S_1$}, this is written as $\mathcal{S}\{f\}(\bfx)$. This means that the function $f$ is transformed by $\mathcal{S}$ from $S_1$ to $S_2$. The result is in the function-space $S_2$. The argument $\bfx$, which is in the domain of $S_2$, is applied to the transformed function. 

There are two major reasons why operator notation is useful. 
First, it is similar to linear matrix form (beyond the scope of this paper), 
which is commonly used in image reconstruction to 
invert measured data (see \cite{zibulevsky2010l1, mccann2019biomedical,zalev2018image,provost2009application, claerbout2008basic}). 
For example,
 $\scrH$ would directly correspond to a discretized system matrix $H$, 
which 
enables solving inverse problems with 
linear algebra. 
In a similar manner,
 the cascade $\tilde\scrG \circ \scrG$ corresponds to the matrix multiplication $\tilde G G$, where $\tilde G$ and $G$ are discretized matrices. 
This can offer computational advantages \cite{zalev2020convex} when using standard inversion methods. 
Secondly, for mathematical proofs, our notation succinctly permits manipulation of continuous variables,   
which are 
hidden in linear matrix notation. 
This simplifies  mathematical descriptions used in this paper and future work involving image reconstruction\cite{zalev2020convex}.

\hfill

\subsubsection{Convolution}

The acoustic signal model involves convolution. 
The convolution of function $f(\bfx)$ with function $g(\bfx)$, in the vector variable $\bfx$, is defined as 
\begin{align}
\label{def:conv}
[f \conv{} g](\bfx) 
&\equiv 
\int_{-\infty}^{\infty}\!
f(\bfx')g(\bfx-\bfx')
\d {\bfx'}.
\end{align}
Here, $\d{\bfx'}$ stands for $\d{x_1'}\d{x_2'}\d{x_3'}$ and  
implies that the integration is carried out over a triple integral $\int_{-\infty}^\infty \int_{-\infty}^\infty \int_{-\infty}^\infty$ due to the three components of $\bfx'$. 
The primed index indicates
that the integration is performed over an intermediate variable $\bfx'$ which is different from (and not to be confused with) the argument variable $\bfx$. 
An intermediate vector variable's elements also have primed indices (e.g., $x_1'$, $x_2'$ and $x_3'$).
If convolution is performed in an expression that  already has one intermediate variable (e.g. $\bfx'$), then a double-primed index is used to indicate a second intermediate variable (e.g. $\bfx''$). 

Often, convolution will be performed over a function of several variables, but the integration will only apply to a subset of these variables. This is indicated by placing the name of the applicable integration variable(s) below the convolution symbol. For example, an expression may read \mbox{$h(\bfx)= f(\bfx) \conv{\bfx} g(\bfx,t)$}. 
In this case, when a replacement for $\bfx$ occurs in the left-hand side $h(\bfx)$, then the right-hand side $f(\bfx) \conv{\bfx} g(\bfx,t)$
should be interpreted 
as in 
 \eqref{def:conv}. 
An explicit replacement uses lower subscripts on square brackets. 
For example, 
\begin{align*}
\left[f(\bfx) \conv{\bfx} g(\bfx) \right]_{\bfx=A\bfu+\bfb}
\end{align*}
is equal to
\begin{align*}
\int_{-\infty}^{\infty}\!
f(\bfx')g(A\bfu+\bfb-\bfx')
\d {\bfx'}.
\end{align*}

As another shorthand, the expression $f(\bfx) \!\conv{x_1,x_2} \!g(\bfx)$ indicates that the convolution is carried out only over $x_1$ and $x_2$. This avoids having to include a factor of Dirac delta (e.g. $\delta(x_3)$) to cancel the variable that is not integrated over.

\subsection{Coordinate frames of linear array} 
\label{sec:coordinate}
In this section we describe the
coordinate frames 
and geometry for the opto-acoustic probe
used in our model.

Figure~\ref{fig.probe_geom}  shows an opto-acoustic probe with linear array in a rotated coordinate frame. 
The coordinate frame of the probe is $\sfF_\text P$. 
The probe (and thus its coordinate frame) is located at position $\bfb$, relative to the base coordinate system $\sfF_0$. 
The advantage of this setup is that the 3D volume of acoustic sources can be maintained in 
$\sfF_0$, which facilitates computation. 
Light is delivered from $\sfF_\text P$, so the optical energy distribution remains stationary relative to the probe, but is translated and rotated relative to the volume. 

The vector $\bfx \in \R^3$ is used to describe a position in the global frame $\sfF_0$, and 
the vector $\bfu \in \R^3$ describes a position in the local frame of the probe $\sfF_\text{P}$. 

The probe orientation is represented by direction vectors $\hat\bfn_1$, $\hat\bfn_2$, and $\hat\bfn_3$ that point along the axes of the probe. 
In a linear array, the elements of the array are arranged on a line pointing along the  
$\hat\bfn_1$ axis. 
Therefore, in frame $\sfF_0$, the line of elements is given by the equation
\begin{align}
\label{eq:line_eq}
\bfx = \hat\bfn_{1} u_1 + \bfb, 
\end{align}
where $u_1$ is a positive or negative distance from $\bfb$ in the $\hat\bfn_{1}$ direction. 
In frame 
 $\sfF_\text P$, a position on the 
line of transducer elements corresponds to $\bfu = (u_1, 0,0)$. 
The \mbox{$k$-th} element of the array is located at $u_1 = (k-1) \Delta_{u_1} + u_0$, where the element spacing is $\Delta_{u_1}$,  
and the first element's position corresponds to  $u_0$.

\begin{figure}
\ifdefined\SHORTTEX
    \centerline{\includegraphics[width=1.0\columnwidth, trim=10 20 10 30, clip]{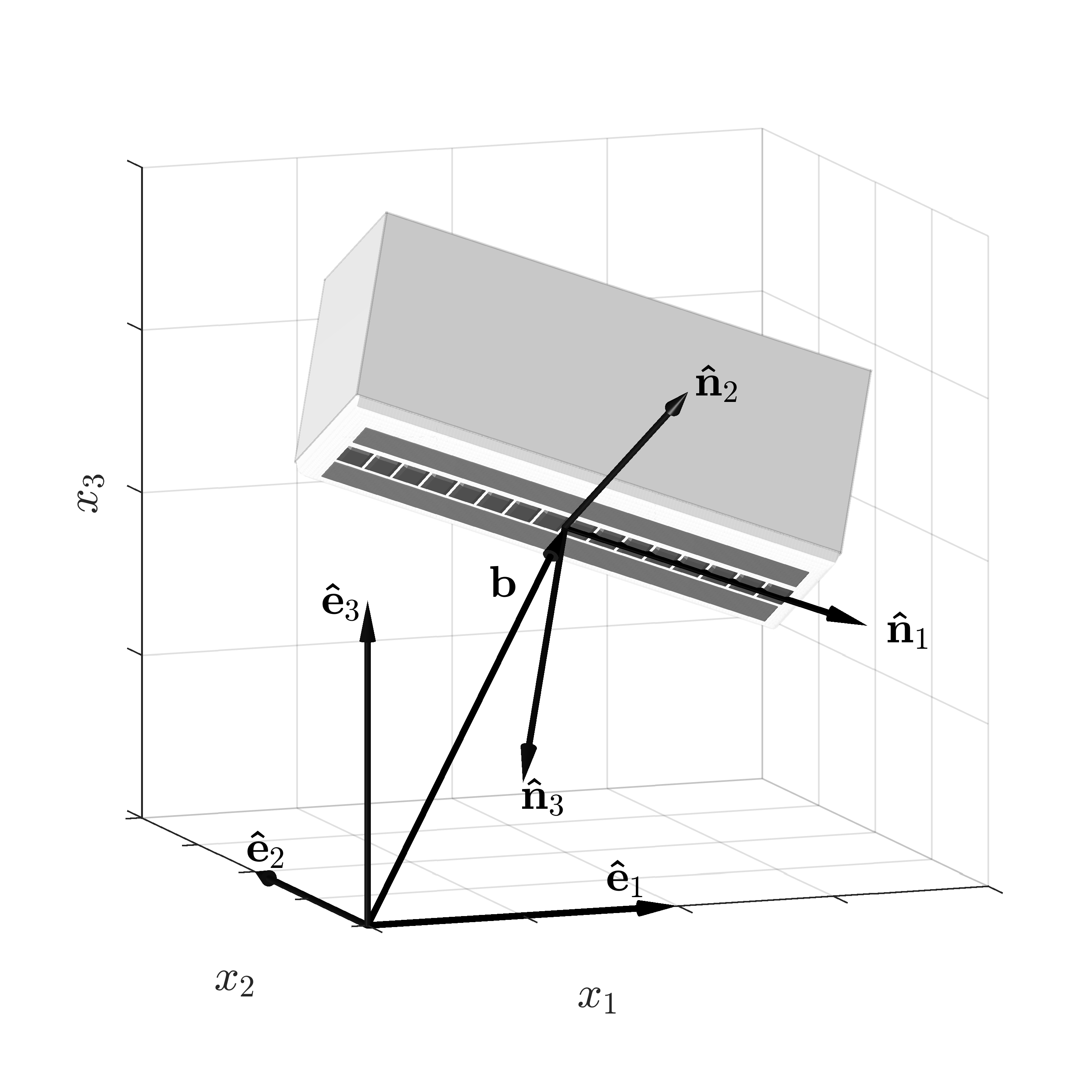}}
\else
    \centerline{\includegraphics[width=0.7\columnwidth, trim=10 20 10 30, clip]{fig_probe_1.png}}
\fi
    \centering
\caption[Geometry of an opto-acoustic probe in rotated and translated coordinate system]{Geometry of an opto-acoustic probe in rotated and translated coordinate system.
A linear array of 16 transducer elements (small dark rectangles on the probe face) is shown along the $\hat\bfn_1$ axis. 
Two rectangular optical source apertures (long dark rectangles on the probe face) are shown 
adjacent
to the array, which illustrate how light is delivered to the volume in a typical scenario. 
The basis vectors for the probe's local coordinate frame $\Fp$, are $\hat\bfn_1$, $\hat\bfn_2$, and $\hat\bfn_3$. 
For the global coordinate frame $\Fo$, the standard basis vectors $\hat\bfe_1$, $\hat\bfe_2$, and $\hat\bfe_3$ are used. 
The probe is translated from the origin by $\bfb$, and rotated by 
the rotation matrix $A=[\hat\bfn_1\quad\hat\bfn_2\quad\hat\bfn_3]$. 
}
\label{fig.probe_geom}
\end{figure}

We consider a rigid coordinate transform from $\sfF_\text{P}$ to $\sfF_0$ using an affine transformation. 
Let $A \in \R^{3\times 3}$ be a rotation matrix, and $\bfb \in \R^3$ be a spatial displacement vector.  
The function $T_{A,\bfb}(\bfu)$ applies the rotation  $A$ to $\bfu$, followed by a translation of $\bfb$.  Thus, 
conversion between the local coordinates $\bfu$ of the probe  and global coordinates  $\bfx$ of the volume  can be performed by
\begin{subequations}
\label{eq:atrans}
\begin{align}
\bfx={T}_{A,\bfb} (  \bfu ) &= A\bfu+\bfb \label{eq:atrans_fwd}
\\
\bfu={T}^{-1}_{A,\bfb} (  \bfx ) &= A^{T}(\bfx-\bfb). \label{eq:atrans_inv}
\end{align}
\end{subequations}

The rotation matrix $A$ for converting orientation from $\sfF_P$ to $\sfF_0$ is given by
\begin{align}
A 
=
\begin{bmatrix}
a_{11} & a_{12} & a_{13} \\
a_{21} & a_{22} & a_{23}  \\
a_{31} & a_{32} & a_{33} 
\end{bmatrix}
= 
\begin{bmatrix}
\uparrow&\uparrow&\uparrow\\
\hat\bfn_{1}& \hat\bfn_{2}& \hat\bfn_{3} \\  
\downarrow&\downarrow&\downarrow
\end{bmatrix}
.
\end{align}
The columns of the rotation matrix correspond to the direction vectors $\hat\bfn_1$, $\hat\bfn_2$, and $\hat\bfn_3$ for the probe orientation. 
The direction vectors form an orthonormal coordinate system, and are related by
$\hat\bfn_1 \cdot\, \hat\bfn_2 =  \hat\bfn_2 \cdot\, \hat\bfn_3 = \hat\bfn_3 \cdot\, \hat\bfn_1 =0$ and $\hat\bfn_3 = \hat\bfn_1 \times\, \hat\bfn_2$.
Since the columns are orthogonal unit vectors, the matrix $A$ is orthogonal; therefore, its inverse $A^{-1}$ is equal to its transpose $A^T$, which is a property of rotation matrices. 
The translation $\bfb$ is the position of the probe in $\sfF_0$, as shown in Figure~\ref{fig.probe_geom}. 

At any position $\bfu$ in the local coordinate frame of the probe, a position $\bfx$ in volume coordinates can be found from Equation \eqref{eq:atrans_fwd}. 
To convert to local coordinates from volume coordinates, the inverse transform of Equation \eqref{eq:atrans_inv} is used.

There is a special situation for a linear-array.
When $\bfu=(u_1, 0,0)$, which means it is constrained to the $\hat\bfn_1$ axis, then $T_{A,\bfb}\left([u_1, 0,0]^T\right)$ reduces to Equation \eqref{eq:line_eq}, which is the equation of a line in $\sfF_0$. Also, applying $T_{A,\bfb}$ to the plane $\bfu=(u_1, 0,u_3)$ will produce a plane in $\sfF_0$ corresponding to the \emph{imaging plane} cross section of the volume. 
The imaging plane is what would be seen on the display of a 2D imaging device when the probe is in position over an imaging volume. 
 When $\bfu=(0, u_2,u_3)$, this corresponds to a plane perpendicular to the imaging plane in $\sfF_0$.

\hfill 

The affine transform function $T_{A,\bfb}$, acts on a vector and generates a vector. 
This is different from the affine transform operator $\scrT_{A,\bfb}$, which acts on a function and produces a rotated and translated version of that function. In certain proofs we rely on one or the other, so we are careful to point out the distinction. 

The affine transform operator (and its inverse), 
are defined by
\begin{subequations}
\label{eq:fwd_inv_affine}
\begin{align}
\scrT_{A,\bfb}\{\psi\}(\bfx)&=
\big[\psi \circ T_{A,\bfb}^{-1}\big](\bfx)
=\psi(A^T(\bfx-\bfb)),
\\
\scrT_{A,\bfb}^{-1}\{\psi\}(\bfu)&=
\big[\psi \circ T_{A,\bfb}\big](\bfu)
=\psi(A\bfu+\bfb).
\end{align} 
\end{subequations}

If a function $\psi(\bfx)$, which describes a spatial distribution, is shifted \emph{in one direction}, this corresponds its argument $\bfx$, which is a coordinate vector, being shifted \emph{in the opposite direction}. 
Accordingly, $\scrT$ is defined by composition with the inverse of $T$. 
In this sense, $T$ and $\scrT$ behave oppositely.

\hfill

In our model, the probe orientation is specified with a rotation matrix. 
Arbitrary rotation matrices can be formed by applying a sequence of 
 Euler matrices corresponding to 
\emph{roll}, \emph{pitch} and \emph{yaw} angles. 
To generate an arbitrary rotation matrix $A$, 
we use the convention
\begin{align}
A=R_1(\theta_1)R_2(\theta_2)R_3(\theta_3), 
\end{align}
where the matrices for roll, pitch and yaw are 
\begin{subequations}
\begin{align}
R_1(\theta_1) &=
\begin{bmatrix}
1 & 0 & 0\\
0 & \cos(\theta_1) & -\sin(\theta_1)  \\
0&  \sin(\theta_1) & { }\cos(\theta_1)  
\end{bmatrix},&\text{(roll)}
\\
R_2(\theta_2) &=
\begin{bmatrix}
\cos(\theta_2) & 0 & \sin(\theta_2) \\
0 & 1 & 0 \\
-\sin(\theta_2) &0& \cos(\theta_2)  \\
\end{bmatrix},&\text{(pitch)}
\\
R_3(\theta_3) &=
\begin{bmatrix}
\cos(\theta_3) & -\sin(\theta_3) & 0 \\
\sin(\theta_3) & \cos(\theta_3) & 0 \\
0 & 0 & 1
\end{bmatrix}.&\text{(yaw)}
\label{eq:R3_yaw}
\end{align}
\end{subequations}

\hfill

In subsequent sections, we convert between coordinates 
in equations involving both convolution and composition.  
The following formula for changing variables is useful when working with an impulse response $h(\bfu)$ defined in a local coordinate system, and a spatial distribution $\psi(\bfx)$ defined in a global coordinate system. 
The coordinate $\bfx$ is related to the coordinate $\bfu$ by 
$\bfx=T_{A,\bfb}(\bfu)$.
For convolution of the functions $h(\bfu)$ and $\psi(\bfx)$, a change of coordinates corresponds to the relationship
\begin{align*}
\left[\left(h\circ T_{A,\bfb}^{-1}\right)\circledast\psi\right]\!(\bfx) =
\frac {1}{\abs{\det A}} \Big[h\circledast \left(\psi \circ T_{A,\bfb}\right) \Big](\bfu).
\end{align*}
This can be seen by writing out the convolutional expression $h(\bfu)\circledast\psi(\bfx)$ in expanded form with $\bfu=T^{-1}_{A,\bfb}(\bfx)$. The differential element $\d\bfu$ is scaled by the determinant of $A$ because $\d\bfu=\d{(A^{-1}(\bfx-\bfb))}=\abs{\det A}^{-1}\d\bfx$. 
For a rotation matrix, since $\det A$ always equals $1$, 
the relationship we are interested in simplifies to
\begin{align}
\label{eq:change_var_conv_comp}
\Big[\scrT_{A,\bfb}\{h\}\circledast\psi\Big](\bfx) =
\Big[h\circledast  \scrT_{A,\bfb}^{-1}\{\psi\}\Big](\bfu),
\end{align}
where $\bfx=T_{A,\bfb}(\bfu)=A\bfu+\bfb$. 

Also, if no rotation is performed, and the coordinates are in the same frame, this will reduce to the (more well known) translation property of convolution, which is 
\begin{align*}
\Big[\scrT_{I,\bfb}\{h\}\circledast\psi\Big](\bfx) =
\Big[h\circledast  \scrT_{I,\bfb}\{\psi\}\Big](\bfx),   
\end{align*} 
where $I$ is the identity matrix, so $\scrT_{I,\bfb}$ is a pure translation. The interpretation is that a translation can go with either 
function in the convolution, when compared in the same coordinate frame. 

\subsection{Opto-acoustic signals }
\label{sec:oasig}

The opto-acoustic signal generation model is now described. 
We develop a generalized \emph{system response operator}, that includes: i) transducer aperture, ii) light distribution,  
iii) probe position, and iv) probe orientation.  
In our notation, 
the operator
 is written $\scrH_{f,\PHI}^{A,\bfb}\{\PSI\}(\bfu,t)$.
The most simple way to think of this is
 as a \enquote{function} that outputs a time-domain signal (in the variable $t$),
for a transducer centered at position $\bfu$ in the probe's coordinate frame. 
Here, $A$ and $\bfb$ specify the probe orientation and position. 
The function $f$ specifies the aperture geometry,  
and  $\PHI$  specifies
optical energy distribution. 
The 3D volume is described by $\PSI$.

The system response operator is developed in three parts.  
First, 
we model a 
system response operator 
$\scrH_f$ in a local coordinate frame, 
 and generalize this to transformed coordinates with $A$ and $\bfb$. 

Next, 
 we describe opto-acoustic wave propagation in more detail and explain the physics for the transducer impulse response.  
This will provide context for our derivations in the following sections. 

Finally, 
we build on this to include light distribution $\PHI$, emitted from the coordinate  frame of a probe.

\hfill

\subsubsection{Spatio-temporal impulse response}
\label{sec:spatiotemporal_impluse}

Heating a medium with an optical pulse at time $t=0$ causes an \emph{initial acoustic source distribution} $\psi(\bfx)$, 
which propagates as an acoustic wave.
In our notation, the \emph{simulated opto-acoustic signal} $s(\bfx,t)$,  measured by an ideal transducer with aperture $f$,
at time $t$, and position $\bfx$, 
  is written
\begin{align}
s(\bfx, t) = 
\scrH_{f}\{\psi\}(\bfx, t). 
\end{align}

The operator $\scrH_f$ 
performs 
convolution of 
  $\psi(\bfx)$ with 
 $h_f(\bfx, t)$. 
 Thus, 
\begin{align}
\label{eq:scrH_bfx_t_def}
\scrH_f\{\psi\}(\bfx, t) &= h_f(\bfx, t) \conv{\bfx} \psi(\bfx).
\end{align}
The 
 function $h_f(\bfx, t)$ is the
\emph{pressure impulse response for aperture $f$}.  
It is defined by convolution of  $f(x)$ with 
$h(\bfx, t)$, according to
\begin{align}
\label{eq:h_conv_f}
h_f(\bfx, t) = h(\bfx,t)\conv{\bfx}f(\bfx).
\end{align}

Here, $h(\bfx,t)$ is 
the 
\emph{free-space pressure impulse response}
given by 
\begin{align}
\label{eq:h_bfx_t_def}
h(\bfx, t) &= 
\frac{\partial}{\partial t}
\frac{\delta\!\left(c_0 t-{\norm \bfx}\right)}{4\pi  t}
,
\end{align}
which is described 
in the next subsection. 
 The speed of sound $c_0$ is assumed to be constant, and $\delta$ is the Dirac impulse function. 

When $f$ is an aperture function, it is defined
\begin{align}
f(\bfx)=\begin{cases}
1, &\text{if $-\bfx$ is on transducer's surface},\\
0, &\text{otherwise}.
\end{cases}
\end{align}
A weighted aperture could also be used (where the value on the transducer surface depends on $\bfx$), but this is not 
implemented in this work. 

The aperture function describes the geometry of the transducer elements. 
A rectangular transducer has a length $a_1$ and width $a_2$.  
The \emph{rectangular transducer aperture function},
where $\bfa=(a_1,a_2)$,
is defined as
\begin{align}
f_\bfa(\bfx)=\rect\left(\frac{x_1}{a_1},\frac{x_2}{a_2}\right)\delta(x_3),
\end{align}
where 
\begin{align*}
\rect\left(x_1,x_2\right)
=\begin{cases}
1, & \abs{x_1} \le \frac 1 2 \text{ and } \abs{x_2} \le \frac 1 2,\\
0, &\text{otherwise.}
\end{cases}
\end{align*}
When the transducer is an \emph{ideal point detector}, the aperture is $f_\delta(\bfx)=\delta(\bfx)$.
In this case, the system response operator can be written
$\scrH\{\psi\}(\bfx,t)$ (with no subscripts), since \eqref{eq:h_conv_f} yields $h_{f_\delta}(\bfx, t)=h(\bfx, t) \conv{\bfx} \delta(\bfx) = h(\bfx, t)$. 

\hfill

 To provide 
 transducer responses
with a directional component,
we define
a \emph{modified system response operator} (written using a bar symbol)
\begin{align}
\label{eq:bar_scrH_f_defn}
 \bar\scrH_f\{\psi\}(\bfx, t) &= \bar h_f(\bfx, t) \conv{\bfx} \psi(\bfx),
\end{align}
where
 \begin{align}
 \label{eq:bar_h_f_defn}
\bar h_f(\bfx, t) = \big[
\underbrace{ h(\bfx,t)\alpha(\bfx)}_{\bar h(\bfx, t)}\big]\conv{\bfx}f(\bfx).
\end{align}
The \emph{modified pressure impulse response} $\bar h(\bfx,t)$ is spatially weighted by $\alpha(\bfx)$.
The function $\alpha(\bfx)$, 
which is called the \emph{obliquity factor}
\cite{lasota1982application,
selfridge1980theory,
stepanishen1971transient,
hunt1983ultrasound}, 
describes the directional response to waves arriving from different angles. 
It
is defined by 
\begin{align}
\label{eq:alpha_defn}
\alpha(\bfx)=\frac{\bfx\cdot\hat\bfn}{\norm\bfx},
\end{align}
where $\hat\bfn$ is the outward normal to 
the transducer 
at $\bfx$. 
In our notation, the bar symbol (over the system operator and impulse responses) indicates that
the obliquity factor
$\alpha(\bfx)$
is included. 
When it is present, 
a transducer measures
the unbalanced component of force, normal to the transducer's surface.
Otherwise, 
the 
transducer measures
the excess acoustic pressure. 
The choice depends on the type of transducer that is being modeled. For a soft-baffle transducer, which is constrained to receive axial forces, 
the obliquity factor is present; for a rigid-baffle transducer, which receives omni-directional forces, it is not present. 

In the sections below we will show a proof for both types of transducer, with and without obliquity factor. 
When either case is relevant, this will be made clear from the context.

\hfill

It is convenient to include position and orientation in the notation of the system response operator. 
This is because when dealing with rotation, both the aperture and probe must be rotated simultaneously, 
which follows from \eqref{eq:change_var_conv_comp}.
 (If only $\bfu$ is rotated, then $\scrH_f$ is not applicable because its aperture $f$ is non-rotated). 
Accordingly, the definition of $\scrH_f^{A,\bfb}$ is 
\begin{align} 
\label{eq:scrH_A_b_f_defn}
\scrH_f^{A,\bfb}\{\psi\}(\bfu, t)
&=
\scrH_f \big\{\scrT_{A,\bfb}^{-1}\{\psi  \}\big\}(\bfu, t) \nonumber
\\
&=
\big[\scrH_f \circ \scrT_{A,\bfb}^{-1} \big]\{\psi  \}(\bfu, t) \nonumber
\\&=
\scrH_f  \{\psi \circ T_{A,\bfb} \}(\bfu, t) \nonumber
\\&=
\scrH_f  \{\psi ^{A,\bfb} \}(\bfu, t), 
\end{align}
where we have defined the oppositely rotated and translated volume as
\begin{align}
\psi^{A,\bfb}(\bfu)=
\scrT_{A,\bfb}^{-1}\{\psi\}(\bfu)=\psi\!\left( T_{A,\bfb} (\bfu)\right).
\end{align}
  
  \hfill
  
After the simulated 
system response
 has been computed, a purely temporal \emph{electro-mechanical impulse response} $h_\text{em}(t)$ can be applied to it, if desired, 
to model the system bandwidth. Post-processing with temporal convolution is written
\begin{align}
\label{eq:h_em_conv}
h_\text{em}(t)
\conv{t}
\scrH_f^{A,\bfb}\{\psi\}(\bfx,t).
\end{align}
Equation~\eqref{eq:h_em_conv} 
decouples the temporal response from the geometry of the transducer element's aperture. 
The remainder of this paper assumes an
ideal electro-mechanical impulse response, where $h_\text{em}(t)$ is equal to $\delta(t)$, 
which 
 leaves $\scrH_f^{A,\bfb}$ unaltered and does not affect computation. 
This assumption is used because opto-acoustic signals have broadband frequency content, 
which is due to using nearly instantaneous optical pulses, 
as described in the next section.

\hfill

\subsubsection{Opto-acoustic wave propagation}
\label{sec:oa_wave_prop}

The opto-acoustic physics relevant to our model is now described.  
First, the pressure impulse response function $h(\bfx, t)$ 
in  \eqref{eq:h_bfx_t_def} is 
 derived.
   This corresponds to solving the acoustic wave equation with an instantaneous  heating source term (see Chapter 7 of Morse and Ingard (1968)\cite{morse1968theoretical}; Oraevsky (2014)\cite{oraevsky2014optoacoustic}; and Xu et al. (2006)\cite{xu2006photoacoustic}).
Next, a derivation is provided for  \eqref{eq:bar_scrH_f_defn} that incorporates transducer directionality 
 into the model. 
Some of the details of this section will be used in demonstrating the separable cascade described in Section~\ref{sec:separable}.

In opto-acoustics, a pulse of light rapidly heats optical absorbers in a volume. 
Thermal expansion  due to the heating pulse creates a 
spatial distribution of omni-directional sources 
that radiate acoustic waves\cite{morse1968theoretical}. 
The acoustic waves propagate 
according to the acoustic wave equation for pressure $p(\bfx,t)$, 
given by 
\begin{align}
\label{eq:oa_wave_fcn}
\left(\frac{\partial^2  }{\partial t^2} 
- c_0^2 \nabla^2 \right)\!p(\bfx, t) 
&=
\psi(\bfx)
\frac{\partial}{\partial t}q(t).
\end{align}
Here, the source term 
$\psi(\bfx)
\frac{\partial}{\partial t}q(t)
$
is present only for the short duration of the pulse. 
We have arranged the source term 
as the product of a temporal factor $(\partial/\partial t) q(t)$ and a spatial factor $\psi(\bfx)$, which simplifies the derivation below. 
The function $q(t)$ is the temporal pulse shape, i.e., the time-domain waveform of the laser.
It is assumed that $q(t)$ is sufficiently 
short
so  thermal and stress conferment occurs (i.e.,
thermal conduction and wave propagation in the medium are negligible for the duration of the pulse.) 
This maximizes energy converted to acoustic waves. 
 We also assume $q(t)>0$ for all $t$, and $\int_{-\infty}^{+\infty} \!q(t) \,\d{t}=1$, so $q(t)$ does not impact the total amount of energy delivered, just its shape. 
For an instantaneous pulse, $q(t)$ is equal to $\delta(t)$, 
which is an assumption used in most of this paper for simplification.

The spatial distribution $\psi(\bfx)$ is an instantaneous acoustic source term known as the \emph{initial excess pressure},
which results from optical heating as
 described in the next section (see equation~\eqref{eq:psi_Psi_Phi}).

At an arbitrary position $\bfx$ and time $t$, 
the wave equation from \eqref{eq:oa_wave_fcn} has a solution 
for pressure $p(\bfx, t)$ of 
\begin{align}
\label{eq:pressure_soln_int}
p(\bfx, t) = g(\bfx, t) 
\conv{\bfx,t} 
\psi(\bfx)\frac{\partial}{\partial t}q(t),
\end{align}
where 
\begin{align}
\label{eq:g_defn}
g(\bfx, t) = \frac{\delta\!\left(c_0 t- \norm\bfx \right)}{4\pi t} . 
\end{align}
The function $g(\bfx,t)$
is known as the causal 
Green's function solution to the three-dimensional d'Alembert wave equation in free-space (a.k.a. the \emph{delayed free-space Green's function})\footnote{
The Green's function $g(\bfx, t)$ can appear in several 
equivalent forms.  
This is because: i) the impulse occurs with $\norm\bfx= c t$; 
and ii) the scaling property of the Dirac impulse function obeys the identity $\delta(t-x/c)\equiv c\delta(ct-x)$. Therefore,
$
g(\bfx, t) 
\equiv \frac{1}{4\pi\norm\bfx}\delta\!\left(t- \frac{\norm \bfx}{c} \right)
\equiv \frac{1}{4\pi c t}\delta\!\left(t- \frac{\norm \bfx}{c} \right)
\equiv \frac{1}{4\pi t}\delta\!\left(c t- {\norm \bfx} \right)
\equiv \frac{c}{4\pi\norm\bfx}\delta\!\left( c t- {\norm \bfx} \right)
$. 
In addition, iii) some authors divide the wave equation (Equation \eqref{eq:oa_wave_fcn})  on both sides by $c^2$, 
in which case $g(\bfx,t)$ would appear with an additional factor of $c^{-2}$ to conserve units.
We are concerned with the pressure impulse response $h(\bfx,t)$, which is proportional to the time-derivative of $g(\bfx,t)$.  
In cases where this time-derivative is evaluated analytically, 
it is convenient to use a form of $g(\bfx, t)$ where $t$ is in its denominator  (instead of $\norm\bfx$), 
to make it apparent that 
the dependency of $\norm\bfx$ on $t$
 requires using the product rule.
However, for evaluating a spatial gradient of $g(\bfx, t)$, 
this is most apparent
in forms with $\norm\bfx$  in the denominator. 
},
which permits Equation \eqref{eq:pressure_soln_int} to be written in convolutional form. 

Since the source term of \eqref{eq:pressure_soln_int} is multiplicatively separable in $\bfx$ and $t$, 
it
can also be written
\begin{align*}
p(\bfx,t)=
			  g(\bfx, t) 
			  \conv{t}
			   \frac{\partial q(t)}{\partial t}			  
			   \conv{\bfx} \psi(\bfx).
\end{align*}
Due to the commutative-derivative property of convolution (the derivative can be performed on either convolutional factor), the $\frac{\partial}{\partial t}$ can be moved from $q$ to $g$. 
Therefore, $p(\bfx, t)$ is equal to
\begin{align*}
			  \frac{\partial g(\bfx,t)}{\partial t}			  
			  \conv{t}
			  {q}(t)
			   \conv{\bfx} \psi(\bfx).
\end{align*}
By defining the free-space pressure impulse response as 
\begin{align}
\label{eq:h_defn}
h(\bfx, t) = \frac{\partial g(\bfx,t)}{\partial t}, 
\end{align}
this becomes
\begin{align}
p(\bfx, t) =
				h(\bfx, t)			  
			  \conv{t}
			  {q}(t)
			   \conv{\bfx} \psi(\bfx), 
\end{align}
which demonstrates that the pressure resulting from an arbitrary pulse shape 
is a temporal convolution with 
$q(t)$.
For an instantaneous pulse, 
when $q(t)=\delta(t)$, this simplifies to 
\begin{align}
\label{eq:p_eq_h_conv_psi}
p(\bfx, t) =	 
h(\bfx, t)			  
  \conv{\bfx} \psi(\bfx).
\end{align}

When a transducer centered at position $\bfx$ measures pressure by integrating over the aperture $f(\bfx)$,  
the resulting signal is
\begin{align*}
h(\bfx, t)	
 \conv{\bfx} f(\bfx)		  
  \conv{\bfx}  \psi(\bfx), 
\end{align*}
which matches the form of Equation~\eqref{eq:scrH_bfx_t_def} that defines $\scrH_f$. 
To include an arbitrary pulse shape $q(t)$ in the resulting time-domain signal, a temporal convolution with $q(t)$ can be performed after the measured signal has been computed (similar to Equation \eqref{eq:h_em_conv}). 
Since this can be applied in post-processing, and grouped with $h_\text{em}$ of Equation~\eqref{eq:h_em_conv},  we focus only on the case where an instantaneous pulse shape is used. 
Therefore, $q(t)$ is not part of our core model, although the effect can be incorporated if necessary. 
Rapid or instantaneous pulses generate opto-acoustic signals with broadband frequency content. 
Increasing the pulse width causes a low-pass filtering effect in accordance with the frequency response of $q(t)$.

\hfill

The physics for $\bar \scrH_f$, when the obliquity factor $\alpha(\bfx)$ is present,  is now described. 
If $\bfx$ represents position on the transducer surface,
the net component of force density
resulting from the pressure gradient,
acting on the transducer in its inward normal direction $\hat\bfn$,  is
$$
-\nabla p(\bfx,t)\cdot\hat\bfn.
$$
The total force acting on the transducer surface $S$ is thus
\begin{align*}
\int_S 
\nabla p(\bfx,t)\cdot\hat\bfn
\,\d S. 
\end{align*}
This can be written as a spatial convolution with an aperture function $f(\bfx)$. 
Furthermore, by using a temporal convolution, 
it can be applied to  
 an electro-mechanical impulse response 
for pressure $h_\text{em}^\text{p}(t)$,
which describes how the transducer responds to a pressure input. 
This yields the transducer output signal as
\begin{subequations}
\begin{align}
\label{eq:output_sig_p}
 h_\text{em}^\text{p}(t) \conv{t}
\left(
\nabla p(\bfx,t)\cdot\hat\bfn 
\right)
\conv{\bfx} f(\bfx).
\end{align}
Equation~\eqref{eq:output_sig_p} can also be written in terms of the \emph{velocity potential} $\phi(\bfx, t)$, which is related to pressure by $p(\bfx, t)=\varrho_0 \frac{\partial }{\partial t} \phi(\bfx, t)$, 
where  $\varrho_0$ is mass-density. 
In this case, the  electromechanical impulse response for velocity $h_\text{em}^\text{v}(t)$ can be used instead of  $h_\text{em}^\text{p}(t)$. These are related by 
$h_\text{em}^\text{p}(t)=\varrho_0\frac{\partial}{\partial t}h_\text{em}^\text{v}(t)$   
(i.e., the step response for measuring pressure is the impulse response for measuring velocity.) 
Thus, equation~\eqref{eq:output_sig_p} can be written equivalently as 
\begin{align}
\label{eq:output_sig_vpot}
h_\text{em}^\text{v}(t)
\conv{t}
\left(
 \frac{1}{\varrho_0}
\nabla \phi(\bfx,t)\cdot\hat\bfn 
\right)
\conv{\bfx} f(\bfx).
\end{align}
\end{subequations}
By taking the gradient of Equation~\eqref{eq:p_eq_h_conv_psi}, the pressure gradient is equal 
to%
\footnote{
The gradient can be grouped with either convolutional factor, due to the commutative-derivative property of convolution
$\nabla[f(x)\circledast g(x)]=\nabla f(x)\circledast g(x) = 
 f(x)\circledast \nabla g(x) $}
$$
\nabla p(\bfx,t) = 
\nabla h(\bfx, t) \conv{\bfx} \psi(\bfx)
=
 h(\bfx, t) \conv{\bfx} \nabla\psi(\bfx). 
$$ 
Using this with the definition of velocity potential, and from Equation \eqref{eq:h_defn},
we can  write the gradient of the velocity potential in terms of the free-space Green's function $g(\bfx, t)$. Thus 
\begin{align}
\label{eq:grad_g_div_rho}
 \frac{1}{\varrho_0}
\nabla \phi(\bfx,t)
= \nabla g(\bfx, t) \conv{\bfx} \psi(\bfx).
\end{align}
From Equation~\eqref{eq:g_defn}~and~\eqref{eq:h_defn},
and since the impulse occurs at $\norm\bfx=c_0 t$, 
the spatial gradient of $g(\bfx, t)$ is
\begin{align*}
\nabla g(\bfx, t) 
&= 
\sum_i \frac{\partial  g(\bfx, t)}{\partial x_i} \hat\bfe_i
\\&=
\sum_i  \frac{\partial g(\bfx, t)}{\partial t}\frac{\partial t}{\partial x_i} \hat\bfe_i
\\&=
\frac{1}{c_0} \frac{\partial g(\bfx, t)}{\partial t} \sum_i\frac{\partial \norm \bfx}{\partial x_i} \hat\bfe_i
\\
&= \frac{1}{c_0} h(\bfx, t) \frac{\bfx}{\norm \bfx}. 
\end{align*}
Therefore
\begin{align*}
\nabla g(\bfx, t) \cdot \hat\bfn &= \frac{h(\bfx, t)}{c_0} \frac{\bfx\cdot\hat\bfn}{\norm \bfx}
= \frac{ h(\bfx, t)}{c_0} \alpha(\bfx). 
\end{align*}
By applying the dot product of $\hat\bfn$ to both sides of \eqref{eq:grad_g_div_rho}, this gives
\begin{align*}
\frac{1}{\varrho_0}\nabla\phi(\bfx,t)\cdot\hat\bfn 
&=  \frac{ \bar h(\bfx, t)}{c_0}   \conv{\bfx} \psi(\bfx).
\end{align*}
By substitution into equation~\eqref{eq:output_sig_vpot}, 
and by ignoring the 
electromechanical impulse response of the transducer (and its gain),  
the overall transducer output signal is proportional to
\begin{align*}
\overbrace{
\big[
\underbrace{
h(\bfx, t) \alpha(\bfx)}_{\bar h(\bfx, t)}
\big]
\conv{\bfx}f(\bfx)
}^{\bar h_f(\bfx,t)}
\conv{\bfx}\,\psi(\bfx).
\end{align*}
This matches the form of Equation~\eqref{eq:bar_scrH_f_defn}, where $\bar\scrH_f$ is defined. 

By omitting the obliquity factor $\alpha(\bfx)$, this reduces to the form of equation~\eqref{eq:scrH_bfx_t_def}, where the transducer is modeled by integrating pressure over an aperture.

\hfill

\subsubsection{Optical energy delivery}
\label{sec:optical_energy}

In our model, the optical energy distribution 
$\PHI(\bfu)$
remains stationary relative to the local coordinate frame $\Fp$ of the probe. 
The optical energy that gets delivered in
the global coordinate frame $\Fo$ corresponds to rotating and translating  $\PHI(\bfu)$
according to the position and orientation of the probe. 
As shown in Figure~\ref{fig.probe_geom}, the probe is positioned at $\bfb$ and rotated by 
$
A
$. 
The relationship between local coordinate $\bfu$ and global coordinate $\bfx$ is
$$\bfu =  T_{A,\bfb}^{-1}(\bfx) = A^T(\bfx-\bfb). $$ 
 Therefore, the light distribution $\PHI(\bfu)$ in frame $\Fo$ corresponds to 
\begin{align}
\PHI_{A,\bfb}(\bfx) 
&= [\PHI(\bfu)]_{\bfu=T_{A,\bfb}^{-1}(\bfx)} \nonumber
\\
&= \PHI\big(T_{A,\bfb}^{-1}(\bfx)\big)
= \scrT_{A,\bfb}\{\PHI\}(\bfx)
.
\end{align}
Here, the rotated and translated optical energy distribution that is delivered to the global frame is called 
$\PHI_{A,\bfb}(\bfx)$.
Notice that applying the affine transformation operator $\scrT_{A,\bfb}$ to $\PHI$, which  
forwardly 
 rotates the distribution by $A$ and translates it by $\bfb$, 
 is equivalent to 
composition of 
 the inverse transformation function $T_{A,\bfb}^{-1}$ with its input. 
  
  \hfill

In opto-acoustics, the  initial excess pressure distribution $\psi(\bfx)$ 
is the product of the opto-acoustic conversion efficiency $\PSI(\bfx)$ times the  optical energy distribution  $\PHI(\bfu)$. 
Thus, 
\begin{align}
\label{eq:psi_Psi_Phi}
\psi(\bfx)&=
{\PSI(\bfx)}\PHI_{A,\bfb}(\bfx).
\end{align}

The opto-acoustic conversion efficiency $\PSI$ describes the amount 
of light 
per unit energy 
that 
is
converted into acoustic pressure 
at each location in the medium. 
It is equal to the optical absorption coefficient $\mu_\text{a}$ 
times the Gr\"uneisen parameter $\Gamma$, both which may vary spatially in the medium.
Therefore, $$\PSI=\Gamma\mu_\text{a}.$$
The coefficient $\mu_\text{a}$ determines how much optical energy is absorbed and converted to heat.  
The pressure increase due to thermal expansion per unit energy gained is determined by $\Gamma$. 
It is equal to
\begin{align*}
\Gamma  = \frac{\beta}{\kappa\varrho_0 c_V} 
=\frac{\beta c_0^2}{c_P},
\end{align*}
which depends on the speed of sound $c_0$, specific heats $c_P$ and $c_V$, mass density $\varrho_0$, compressibility $\kappa$, and thermal expansion $\beta$ coefficients. 

\hfill
 
To facilitate including  light distribution in our model, we extend the definition of the system response operator $\scrH$.
Recall from Equation~\eqref{eq:scrH_A_b_f_defn}, the affine transformed system response for $\psi$ is 
\begin{align*}
\scrH_f^{A,\bfb}\{\psi\}(\bfu,t)&=\scrH_f\big\{\scrT^{-1}_{A,\bfb}\{\psi\}\big\}(\bfu,t)
\\&=
\scrH_f\{\psi^{A,\bfb}\}(\bfu,t), 
\end{align*}
 where  
 $\psi^{A,\bfb}(\bfu)=\scrT_{A,\bfb}^{-1}\{\psi\}(\bfu)$. 
The operator 
applies an inverse affine transform operator to $\psi$, or equivalently stated, it operates directly on $\psi^{A, \bfb}$.
From Equation~\eqref{eq:psi_Psi_Phi},  in frame $\Fo$, the acoustic source distribution $\psi$ is given by the product
\begin{align*}
\psi(\bfx)&=\PSI(\bfx)\,\PHI_{A,\bfb}(\bfx)
\\&=\big[\PSI\odot\PHI_{A,\bfb}\big](\bfx).
\end{align*}
The symbol \enquote{$\odot$} represents  multiplication, which is applied pointwise on the two functions%
\footnote{For discretized functions, the  $\odot$ operation  corresponds to a Hadamard product.
Support for this is directly built in to the syntax of some high-level languages, including Matlab.   
}. 
Thus, $$\psi=\PSI\odot\PHI_{A,\bfb}.$$
Therefore, by applying $\scrT_{A,\bfb}^{-1}$ to both sides,
the acoustic source distribution in the local frame $\Fp$, denoted $\psi^{A,\bfb}$, is
\begin{align*} 
\psi^{A,\bfb}(\bfu)
&=
\scrT_{A,\bfb}^{-1}\left\{
\PSI\odot \scrT_{A,\bfb}\{\PHI\}
\right\}\!(\bfu)
\\&=
\scrT_{A,\bfb}^{-1}\big\{
\PSI\odot\PHI_{A,\bfb}
\big\}(\bfu)
\\&=
\PSI\big(T_{A,\bfb}(\bfu)\big)\,\PHI_{A,\bfb}\big(T_{A,\bfb}(\bfu)\big)
\\&=\PSI^{A,\bfb}(\bfu)\,\PHI(\bfu) 
\\&=\big[\PSI^{A,\bfb}\odot\PHI\big](\bfu), 
\end{align*}
where $\PSI^{A,\bfb}(\bfu)=\scrT_{A,\bfb}^{-1}\{\PSI\}(\bfu)$.
In summary, 
the relationship 
for the acoustic source distribution in either frame
is
\begin{subequations}
\begin{align}
\psi(\bfx)&=\big[\PSI\odot\PHI_{A,\bfb}\big](\bfx),
\\
\psi^{A,\bfb}(\bfu)&=\big[\PSI^{A,\bfb}\odot\PHI\big](\bfu). 
\end{align}
\end{subequations}

Here, because the optical energy distribution can be represented in the local frame $\PHI(\bfu)$, it is independent of the position and orientation of the probe. Therefore, 
to include optical energy distribution in the model, the system response operator 
must perform a pointwise multiplication by $\PHI$, which remains the same even if the probe moves. 

Accordingly, the system response operator that includes the optical energy distribution is defined as
\begin{align}
\label{eq:scrH_ABfPhi}
\scrH_{f, \PHI}^{A,\bfb}\{\PSI\}
&= 
\scrH_{f}^{A,\bfb}\{\psi\}
\nonumber
\\&=
\scrH_{f}\{\scrT_{A,\bfb}^{-1}\{\PSI\}\odot\PHI\}
\nonumber
\\&=
\scrH_{f}\{\PSI^{A,\bfb}\odot\PHI\}.
\end{align}
Here, the optical energy distribution $\PHI$ appears in the subscript of $\scrH_{f, \PHI}^{A,\bfb}\{\PSI\}$
to indicate that 
it is multiplied with
the conversion efficiency $\PSI$ after being rotated.
This notation is convenient because when the probe  is translated and rotated, 
the light distribution stays fixed relative to the probe.

\hfill 

Computing the optical energy distribution $\PHI(\bfu)$ is now briefly described. 
In our model, we assume that regardless of the probe orientation and position, 
$\PHI(\bfu)$ will not change. 
This would be applicable in two situations: a) when the bulk optical properties of the tissue medium are assumed to be homogeneous; and, b) the bulk optical properties of the tissue are layered and the probe moves along the surface, so that by translational symmetry, the motion doesn't alter the optical fluence distribution relative to the probe. 
In both cases, it is assumed that $\PHI(\bfu)$ is not affected by local changes in optical absorption $\mu_a(\bfx)$.

The optical energy distribution is governed by the radiation transport equation\cite{farrell1992diffusion,zemp2013phase, liemert2012light,
liemert2012analytical,
hielscher1998comparison, wang1995mcml,wilson1983monte}.
The number of photons that reach positions in the medium depend on 
its local optical properties, which
change with
 optical wavelength and the configuration of the optical source aperture.  
Light distributions can be solved in several ways, including  diffusion approximations \cite{liemert2012analytical, liemert2012light, farrell1992diffusion, zemp2013phase},
finite-element models \cite{hielscher1998comparison}, 
or  Monte Carlo based methods\cite{wang1995mcml, wilson1983monte}.

\begin{figure}
\ifdefined\SHORTTEX
    \centerline{\includegraphics[width=1.0\columnwidth, trim=15px 35 10 40, clip]{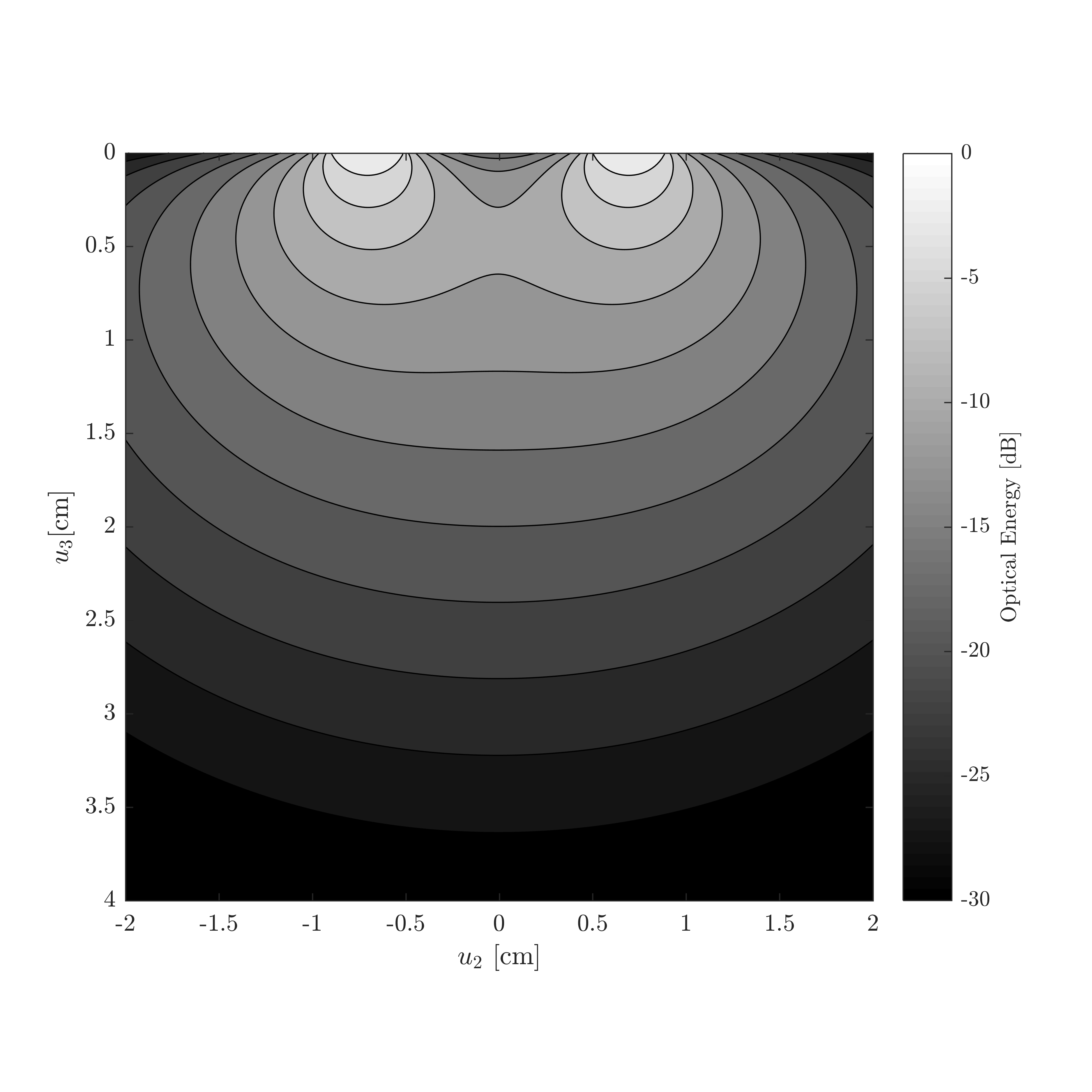}}
\else
    \centerline{\includegraphics[width=1.0\columnwidth, trim=15px 35 10 40, clip]{fig_fluence1.png}}
 \fi
    \centering
    \vspace{-10px}
\caption[Cross section of a typical optical energy distribution]{Cross section of a typical optical energy distribution $\PHI(\bfu)$ for probe in contact with tissue. 
The 2D cross-section represents the $u_2$-$u_3$-plane of the transducer at $u_1=0$.  
The delivered optical energy is highest near each of the optical source apertures located at the tissue surface.
The imaging plane corresponds to a vertical line at $u_2=0$. 
Tissue is modeled with 
effective optical attenuation $\mu_\text{eff}=1.8\text{cm}^{-1}$
for optical wavelength $\lambda=757$nm. 
As the probe moves, $\PHI(\bfu)$ remains fixed relative to the probe coordinate system $\bfu$. 
 }
\label{fig.opticalEnergyDist}
\end{figure}

When delivering light through an optical source aperture, 
the optical energy distribution 
 $\PHI(\bfu)$ can be computed in two steps. 
First,  
an intermediate 
distribution  $\tilde\PHI(\bfu)$ for a pencil-beam (point aperture)  is produced. 
Then, spatial convolution of $\tilde\PHI(\bfu)$ with an optical aperture function yields $\PHI(\bfu)$. 
For a rectangular optical source aperture with dimensions $\bfa=(a_1, a_2)$, centered at the origin,  
the energy distribution is 
\begin{align}
\label{eq:optical_conv}
\PHI(\bfu) = 
\tilde\PHI(\bfu)\conv{u_1, u_2}\rect\left(\frac{u_1}{a_1}, \frac{u_2}{a_2}\right). 
\end{align}

A simple method to compute 
$\tilde\PHI(\bfu)$ 
involves
an empirically derived analytical model for a half-space of optically homogeneous tissue\cite{farrell1992diffusion}. 
The optical energy distribution from a point source
is modeled in terms of a material's \emph{effective optical attenuation} $\mu_\text{eff}$ by 
\begin{align}
\label{eq:optical_energy_point_source}
\nonumber
\tilde\PHI(\bfu) &=
\frac{\mu_\text{eff}^2}{4 \pi }
 \Bigg(
	\frac{\exp{ \left(- \mu_\text{eff} \norm{\bfu - (0,0,z_0) } \right) }}{\norm{\bfu - (0,0,z_0)}}
\ifdefined\SHORTTEX	
\eqbreak
\fi
  - \frac{\exp{ \left( -\mu_\text{eff} \norm{\bfu+ (0,0,z_0+2z_b)}\right)}}
  {\norm{\bfu+ (0,0,z_0+2z_b)}}
  \Bigg),
\end{align}
where
$z_0$ is the  mean-free path length and
$z_b$ is a negative image source distance  
(see Farrell et al. (1992)\cite{farrell1992diffusion}). 
Further enhancements  to include directed photons are described by Zemp (2013)\cite{zemp2013phase}. 

In Figure~\ref{fig.opticalEnergyDist}, the optical energy distribution  
 is shown 
when applied to a probe with two 
rectangular optical source apertures
using spatial convolution of Equation \eqref{eq:optical_conv} with \eqref{eq:optical_energy_point_source}. 
This matches the geometry of Figure~\ref{fig.probe_geom}. 
The delivered optical energy is highest near each of the optical source apertures, which are located at the tissue surface.
As the probe moves, $\PHI(\bfu)$ remains fixed relative to the probe coordinate system $\bfu$. 
Therefore $\PHI(\bfu)$ only needs to be computed once (per optical wavelength), and can then be cached. 
This permits more elaborate methods of computing a fluence distribution to be used in place of 
Equations \eqref{eq:optical_conv} and \eqref{eq:optical_energy_point_source}.

\section{Separable Cascade}
\label{sec:separable}

We demonstrate that 
the system response operator $\scrH^{A,\bfb}_{f,\PHI}$
described in Section~\ref{sec:model} can be computed by a separable cascade. 
Using a separable operator permits an order-of-magnitude reduction in processing time
compared to a non-separable operator.

First, in Section~\ref{sec:separable_A}, we show that the 
system response operator
$\scrH^{A,\bfb}$ 
can be separated
into operator factors, in an arbitrary coordinate frame. 

In Section~\ref{sec:separable_B}, we show that when a multiplicatively separable transducer aperture function $f(\bfx)=f_1(x_1)f_2(x_2)\delta(x_3)$ is used, then the multiplicative factors $f_1$ and $f_2$ can be grouped with the factors of the  system response. 

In Section~\ref{sec:separable_C}, we show that when the obliquity function is used to include directionality, 
the modified operator $\bar\scrH^{A,\bfb}_{f,\PHI}$ is still separable. 

In Section~\ref{sec:separable_D}, we demonstrate that when a rectangular transducer aperture is used, 
a far-field approximation can be used to make the computation even more efficient. 

\hfill

To simplify the expressions in this section, 
we set speed of sound $c_0=1$ and mass-density $\varrho_0=1$.

\begin{figure*}[t]
    \centering
    \subfloat[$  \scrG_{
}^{A,\bfb}\{\psi\}
\rightarrow \sigma(\upsilon_1, 0, \tau) $]{%
    \includegraphics[width=0.5\textwidth, trim=15px 30px 30px 40px, clip]{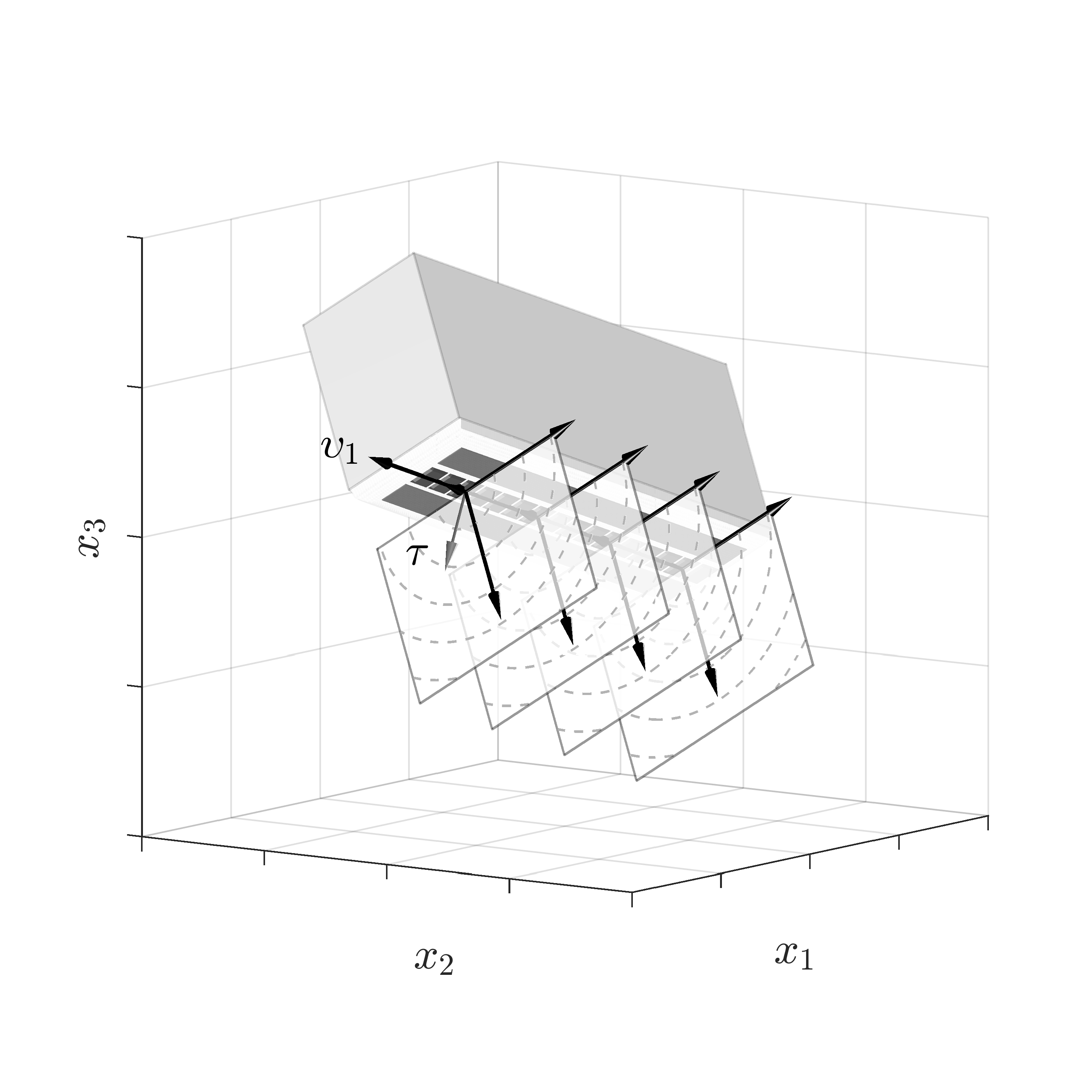}%
}%
\subfloat[ $
   \tilde \scrG
    \{\sigma\}
    \rightarrow 
     s(u_1, 0, t)
    $
]{%
    \includegraphics[width=0.5\textwidth, trim=15px 30px 30px 40px, clip]{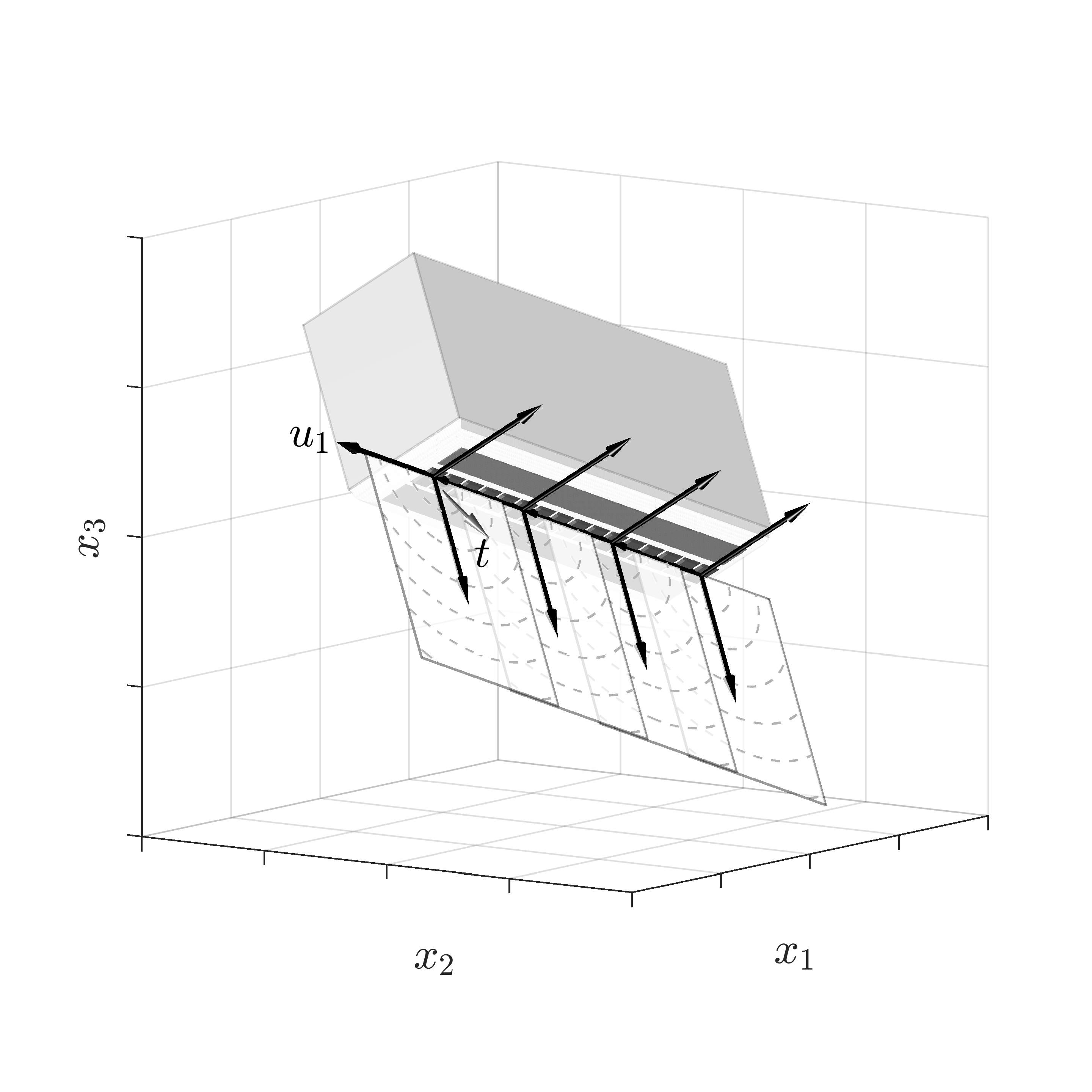}%
    }
\caption[Cascaded separable operators for linear array]{%
Cascaded separable operators for linear array applied to a spatial source distribution $\psi(\bfx)$. 
(a) At each displacement $\upsilon_1$ along the line of transducer elements $\bfx=A[\upsilon_1,0,0]^T+\bfb$, the operator 
$\mathcal G_{
}^{A,\bfb}\{\psi\}
$ 
produces $\sigma(\upsilon_1,0,\tau)$, an intermediate result, by
 performing integration along 
an arc of radius $\tau$, in a plane normal to the line of transducer elements. 
The integration paths span the three-dimensional space $\psi$ (in the planes as shown) by varying $\upsilon_1$ and $\tau$.     
(b) At each displacement $u_1$ in the local coordinate system, 
the operator $\mathcal G_{
}^{\left[R_3^{90\text\textdegree}\!\!,\,\mathbf{0}\right]}\{\sigma\}$ 
produces $s(u_1, 0, t)$ by performing integration along 
an arc of radius $t$ in a plane rotated by 90\textdegree. The integration paths span a two-dimensional subspace of $\sigma$ (shown as the imaging plane), coinciding with the domain $(\upsilon_1,0,\tau)$ of the previous operation.%
}
\label{fig.probe_geom_pll}
\end{figure*}


\subsection{Separability of the system response for acoustic pressure }
\label{sec:separable_A}

In the local probe frame $\Fp$, 
the system response for acoustic pressure, from Equation~\eqref{eq:scrH_A_b_f_defn}, is 
\begin{align*}
p(\bfu, t) = \scrH\{\psi^{A,\bfb}\}(\bfu,t) = h(\bfu,t)\conv{\bfu}\psi^{A,\bfb}(\bfu), 
\end{align*}
where $\psi^{A,\bfb}=\scrT_{A,\bfb}^{-1}\{\psi\}$ is the (oppositely) rotated and translated acoustic source distribution for $\psi$. 

To show that 
this
is separable in an arbitrary coordinate frame, 
in the following proposition
we split
$\scrH^{A,\bfb}$
into the factors $\scrG^{A,\bfb}$ and $\tilde\scrG$. 
 The first factored operator $\scrG^{A,\bfb}$ acts on two-dimensional slices, to span the entire three-dimensional volume 
(see Figure~\ref{fig.probe_geom_pll}a).
The second factored operator $\tilde\scrG$ also acts in two-dimensional slices, 
but these are rotated, so they only span a two-dimensional subspace of the previous output, because all of the transducer elements 
in the array are co-linear (see Figure~\ref{fig.probe_geom_pll}b).

To keep the dimensions equal for each operator factor, and show that both factors involve rotated versions of the same operator, 
we set $u_3=0$ in the definition below. 
Thus, $\scrH^{A,\bfb}$
outputs a function of three arguments $u_1$, $u_2$ and $t$.
However, in a linear array, only $u_1$ and $t$ are relevant, since transducer positions 
are on the line $u_3=u_2=0$.

\begin{proposition}[Separability of acoustic pressure impulse response in transformed coordinates]
\label{prop:s1}
Let the operator $\scrH^{A,\bfb}$ be defined by
\begin{align*}
\scrH^{A,\bfb}\{\psi\}(u_1, u_2, t) &= 
\left[
h(\bfu, t) \conv{\bfu} \psi^{A,\bfb}(\bfu)
\right]_{u_3=0},
\end{align*}
where 
$
h(\bfu, t) =
\frac{1}{4\pi}\frac{\partial}{\partial t}
\frac{\delta\!\left(t-\norm\bfu\right)}{t}
$
and
$\psi^{A,\bfb}=\scrT_{A,\bfb}^{-1}\{\psi\}$. 
Then it can be computed using a separable cascade according to
\begin{align*}
\scrH^{A,\bfb} &= \tilde\scrG \circ \scrG^{A,\bfb},
\end{align*}
where
\begin{align*}
\scrG^{A,\bfb}\{\psi\}(u_1, u_2, t) &= 
 t \! \int_{0}^{2\pi} \! \psi^{A,\bfb} \left( 
\begin{bmatrix}
u_1\\u_2-t \cos \theta \\ -t \sin \theta
\end{bmatrix}
 \right)
 \d \theta,  	
\end{align*}
and
\begin{align*}
\tilde\scrG\{\psi\}(u_1, u_2, t) = 
\frac{1}{4\pi }\frac{\partial}{\partial t}
\left(
\frac{1}{t}
\scrG^{R_3(90^\circ),\bfzero}
\{
\psi
\}
(u_2, -u_1, t)
\right).
\end{align*}
\end{proposition}

\begin{proof}
The proof 
for Proposition~\ref{prop:s1} 
 is demonstrated in \mbox{Appendix~\ref{sec:proof1}}. 
\end{proof}

\subsection{Separability with transducer aperture function}
\label{sec:separable_B}

In this section we show that the 
the operator  $\scrH^{A,\bfb}_f$ retains 
its compositional separability  when a multiplicatively separable aperture function
$f(\bfu)=f_1(u_1)f_2(u_2)\delta(u_3)$
 is used. 
This becomes useful for computing the operator efficiently 
when 
using a rectangular aperture function $f_\bfa(\bfx)=\rect\left(\frac{x_1}{a_1}, \frac{x_2}{a_2}\right)\delta(x_3)$.

Here,  we are interested developing the factors $\scrG_{f_i}$ so that the separable cascade has the form
\begin{align*}
\scrH_f^{A,\bfb}=\tilde\scrG_{f_1 } \circ  \scrG^{A,\bfb}_{f_2}.  
\end{align*}
The subscripts $f_1$ and $f_2$ indicate that a factored 1D aperture function is paired 
with each separable operator factor.

\begin{proposition}[Compositional separability for multiplicatively separable aperture]
\label{prop:sep_conv}
If an aperture function $f(\bfu)$ can be written as a multiplicatively separable function 
$f(\bfu)=f_1(u_1)f_2(u_2)\delta(u_3)$, then 
the operator $\scrH_f^{A,\bfb}$ can be computed by a separable cascade, according to
\begin{align*}
\scrH_f^{A,\bfb}=\tilde\scrG_{f_1 } \circ  \scrG^{A,\bfb}_{f_2}, 
\end{align*}
where 
\begin{align*}
\tilde \scrG_{f_1 }(u_1, u_2, t)&=\tilde\scrG(u_1, u_2, t) \conv{u_1} f_1(u_1), 
\intertext{and}
\scrG^{A,\bfb}_{f_2 }(u_1, u_2, t)&=\scrG^{A,\bfb}(u_1, u_2, t) \conv{u_2} f_2(u_2). 
\end{align*}
\end{proposition}
Here, the operators $\tilde\scrG$ and $\scrG^{A,\bfb}$ are defined as in Proposition~\ref{prop:s1}.

\begin{proof}
The proof for 
Proposition~\ref{prop:sep_conv}
 is provided in \mbox{Appendix~\ref{sec:proof2}}. 
\end{proof}

\subsection{Separability with obliquity factor}
\label{sec:separable_C}

In this section we demonstrate that operator $\bar\scrH_f^{A,\bfb}$ is separable
 when the obliquity factor $\alpha(\bfx)=\frac{\bfx\cdot\hat\bfn}{\norm\bfx}$ is included to obtain a more realistic pressure impulse response. 
From equations~\eqref{eq:bar_scrH_f_defn} and \eqref{eq:scrH_A_b_f_defn}, the modified operator 
in the local coordinate $\bfu$ is
\begin{align*}
 \bar\scrH_f^{A,\bfb}\{\psi\}(\bfu, t) &= 
\big[
{ h(\bfu,t)\alpha(\bfu)}\big]\conv{\bfu}f(\bfu) \conv{\bfu} \psi^{A,\bfb}(\bfu), 
\end{align*}
where
\begin{align*}
\psi^{A,\bfb}(\bfu)=
\scrT_{A,\bfb}^{-1}\{\psi\}(\bfu)
.
\end{align*}
We prove that its computation only requires
adding \mbox{simple} pre- and post-multiplications to the separable operator $\scrH$.

\begin{proposition}[Separability of system response with obliquity factor] 
\label{prop:s3}
The operator $\bar\scrH_f^{A,\bfb}$ can be computed by separable cascade according to 
\begin{align*}
\bar\scrH_f^{A,\bfb}\{\psi\}(u_1, u_2, t)=
\frac{1}{t}
\scrH\{\tilde\psi^{A,\bfb}\}(u_1, u_2, t)\!\conv{u_1, u_2}\!f(u_1, u_2), 
\end{align*}
where
\begin{align*}
\tilde\psi^{A,\bfb}(\bfu) = 
u_3 \,\psi^{A,\bfb}(\bfu).
\end{align*}

\end{proposition}

\begin{proof}
The proof for 
Proposition~\ref{prop:s3}
 is provided in \mbox{Appendix~\ref{sec:proof3}}. 
\end{proof}

\subsection{Far-field approximation to improve performance}
\label{sec:separable_D}

In this section, a far-field approximation to improve performance for computing
$\scrG_{f_i}$ 
is demonstrated for  
rectangular apertures. 
This avoids computing a spatial convolution with the aperture function $f_\bfa(\bfx)$ on the 3D volume $\psi(\bfx)$. It also avoids using 1D spatial convolutions as described in Proposition~\ref{prop:sep_conv}.
The resulting computation involves integration along a path consisting of two positive circular arcs, two negative circular arcs 
and connecting lines, 
 as shown in Figure~\ref{fig.arcpath}b. 
Following this, efficient time-domain filtering operations (integration and multiplication) are performed, which results in the desired signals. 
The far-field approximation becomes more accurate when $\bfx \gg a$. Thus, the approach is applicable when extreme accuracy is not required in the near-field. However, even when used in the near-field, the far-field approximation may provide more realism compared to modeling an ideal point aperture.

\begin{proposition}[Far-field approximation to improve performance]
\label{prop:s4}
Let a separable factor of the system response be defined by 
\begin{align*}
\scrG\{\psi\}
(u_1, u_2, t)
=
\left[
\psi(\bfu)
\conv{\bfu}
\delta(u_1)
\delta(t - \norm\bfu)
\right]_{u_3=0}
,
\end{align*}
and its convolution with a rectangular pulse of width ${a}$ be defined by 
\begin{align*}
\scrG_{a}
\{\psi\}(u_1, u_2, t)
=
\rect\left(\frac{u_2}{a}\right)\conv{u_2}
\scrG\{\psi\}
(u_1, u_2, t).
%
\end{align*}
Then, $\scrG_{a}$ can be computed as 
\begin{align*}
\scrG_{a}
\{\psi\}(u_1, u_2, t)
 = t\int_0^t\!\d{t}\, \frac{\partial}{\partial t}\!
 \left(
 \frac{1}{t}\scrG_{a}\{\psi\}(u_1, u_2, t) \right)
 ,  
\end{align*}
where $\frac{\partial}{\partial t}\left(\frac 1 t \scrG_{a} \right)$ 
is approximated by integrating 
$\psi(z,x,y)$  
over six contours
according to
\begin{align*}
&\frac{\partial}{\partial t}\left(\frac 1 t \scrG_{a}\{\psi\}(u_1, u_2, t)\right)\approx
\\&\quad+
t\int_{\frac{\pi}{2}}^{\frac{3\pi}{2}}
\psi\left(
\begin{bmatrix}
u_1\\ (u_2+a)-t\cos\theta\\ t\sin\theta
\end{bmatrix}
\right)
\frac{1}{\gamma(t, \theta)}\,
\d\theta
\\&\quad-
t\int_{-\frac{\pi}{2}}^{\frac{\pi}{2}}
\psi\left(
\begin{bmatrix}
u_1\\ (u_2+a)-t\cos\theta\\ t\sin\theta
\end{bmatrix}
\right)
\frac{1}{\gamma(t, \theta)}\,
\d\theta
\\&\quad-
t\int_{\frac{\pi}{2}}^{\frac{3\pi}{2}}
\psi\left(
\begin{bmatrix}
u_1\\ (u_2-a)-t\cos\theta\\ t\sin\theta
\end{bmatrix}
\right)
\frac{1}{\gamma(t, \theta)}\,
\d\theta
\\&\quad+
t\int_{-\frac{\pi}{2}}^{\frac{\pi}{2}}
\psi\left(
\begin{bmatrix}
u_1\\ (u_2-a)-t\cos\theta\\ t\sin\theta
\end{bmatrix}
\right)
\frac{1}{\gamma(t, \theta)}\,
\d\theta
\\&\quad+
\int_{-\frac a 2}^{\frac a 2}
\psi\left(
\begin{bmatrix}
u_1\\ u_2-x \\ t
\end{bmatrix}
\right)
\d{x}
\\&\quad+
\int_{-\frac a 2}^{\frac a 2}
\psi\left(
\begin{bmatrix}
u_1\\ u_2-x \\ -t
\end{bmatrix}
\right)
\d{x},
\end{align*}
where
$$
\gamma(t, \theta) = \max\left(\frac a 2, t\abs{ \cos \theta}\right).
$$
\end{proposition}

\begin{figure}
	\centerline{\includegraphics[width=0.5\textwidth, trim=20 3 20 20, clip]{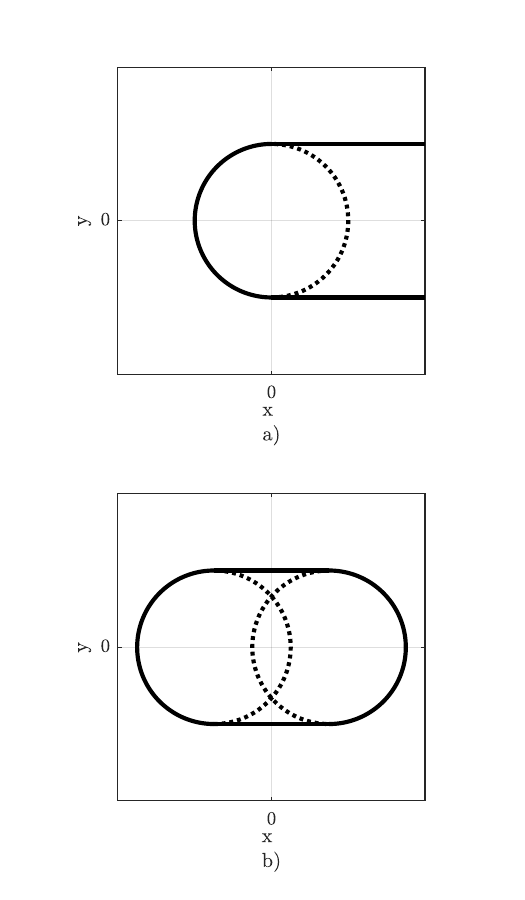}}
    \centering
    \vspace{-10px}
\caption[Spatial impulse responses  in the $z=0$ plane]{Spatial impulse responses  in the $z=0$ plane. Solid line represents positive wavefront and dotted line represents negative wavefront. 
  a) The ray impulse response $\varsigma(x,y,t)$ at time $t$ 
  is $\varsigma(x,y,t)=\frac 1 {2} \delta(t\pm y)\step(x) -
\delta(t-\sqrt{x^2+y^2})\signum(x)$, which 
 traces out the path shown.   
  b) The impulse response $h(x,y,t)$ at time $t$ from a line segment of width $a$. This is formed by superposition of responses from two rays, with 
$h(x,y,t)=\varsigma(x+a/2,y,t)-\varsigma(x-a/2,y,t)$.   
}
\label{fig.arcpath}
\end{figure}

Here, the 2D system response factor operator $\scrG\{\psi\}(u_1,u_2,t)$ has a 3D input domain $\psi$
and produces a 3D scalar output that is parametrized by $(u_1, u_2,t)$. 
It is considered 2D because at a fixed value of $u_1$, the span of its integration is on a 2D plane, as shown in Figure~\ref{fig.probe_geom_pll}a. 
This is also the case for Figure~\ref{fig.probe_geom_pll}b, but where $u_2$ is the fixed parameter (along the line of the array elements, $u_2$ equals $0$). 
\begin{proof}
The proof for 
Proposition~\ref{prop:s4}
 is provided in \mbox{Appendix~\ref{sec:proof4}}. 
\end{proof}

Note that an additional assumption can be made if there is zero signal that arrives from behind the transducer, which is typically desired in practice. 
Instead of six contours, there would only be five (since one of the connecting lines is on the opposite side of the transducer). In addition, each arc would cover an angle of $\frac \pi 2$ rather than $\pi$, since integration does not need to be performed in the half-space behind the transducer.

\section{Results and Implementation}
\label{sec:impl}

A model for opto-acoustic simulation 
using a separable mathematical operator 
 was implemented as described in Section~\ref{sec:model} and Section~\ref{sec:separable}. 
To achieve high performance
for computing $\scrH_{f,\PHI}^{A, \bfb}$, 
a custom GPU (Graphics Processing Unit) kernel was developed to implement
 the separable operator factor $\scrG_{f}^{A,\bfb}$ and its \mbox{adjoint}. 
The GPU kernel was developed using CUDA. Testing involved using a GeForce GTX Titan X graphics processor, which has 3072 parallel cores and 12GB of graphics memory. 
The GPU kernel was configured to run from within a Matlab software environment. 
A user interface was developed to adjust the position and orientation of a linear-array probe relative to a volume of acoustic sources, as illustrated in Figure~\ref{fig.vol3d}. Controls were implemented to adjust roll, pitch and yaw, as well as the geometry parameters for the transducer elements. 
Simulation from a volume consisting of $512\times512\times512$ voxels was capable of achieving real-time performance for generating time-domain data with output of 256 transducer channels and 1024 samples per channel. 
In Matlab, the display was updated at a rate of 5 frames per second, which could be increased by reducing the size of the volume.
Computation times for various grid sizes are provided in Table~\ref{tab:sep_timings}. 

In Figure~\ref{fig.vol3d}, 
the configuration of a 128-channel linear-array and its imaging plane 
are shown, 
relative to a volume consisting of seven spherical opto-acoustic sources.  
The linear-array output is shown in Figure~\ref{fig.sinogram1}, generated from simulation based on Prop~\ref{prop:s1} using the ideal point detector to produce 512 time-domain samples. 
A pictorial representation of the time-domain data, called a \emph{sinogram}, is shown along with time-domain traces from selected channels of the dataset. In our implementation, 
to avoid aliasing, 
the temporal sampling rate $f_s$ is set by default to the critical rate $f_c$. 
This is determined
based on the desired number of samples $n_s$ and the voxel spacing $\Delta x=L_z/n_z$, 
according to 
$f_c=\frac{c_0}{\Delta x}$, 
where $L_z$ and $n_z$ are length and number of voxels along the $z$-axis, and $c_0$ is speed of sound.

\begin{figure}[h]
\ifdefined\SHORTTEX
    \centerline{\includegraphics[width=1.0\columnwidth, trim=20 20 25 40, clip]{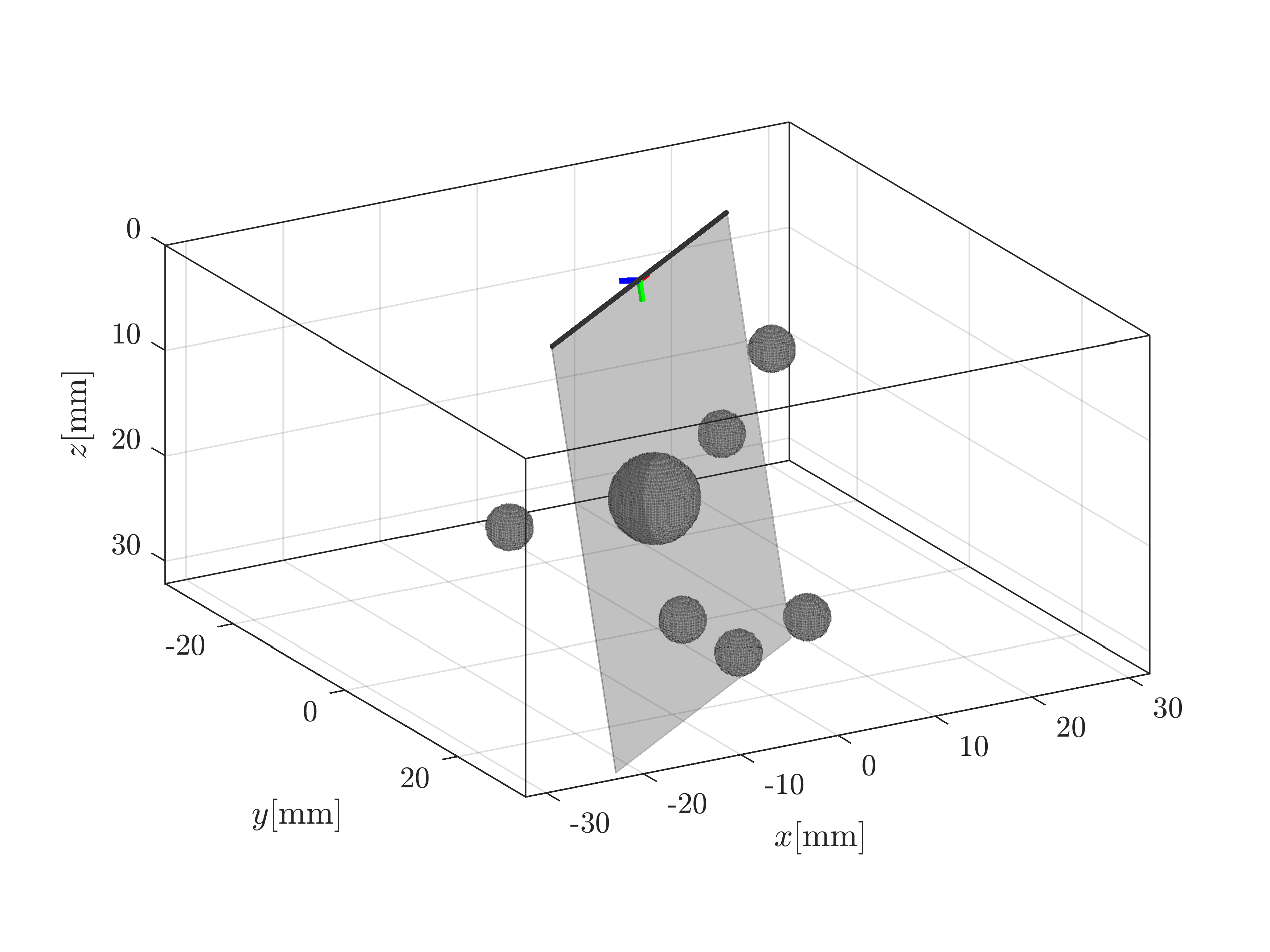}}
\else
    \centerline{\includegraphics[width=0.9\textwidth, trim=20 20 25 40, clip]{fig_volume1.png}}
\fi
    \centering
    \vspace{-10px}
\caption[3D volume consisting of several spherical opto-acoustic sources of equal intensity]{A 3D volume consisting of several spherical opto-acoustic sources of equal intensity. 
Position of linear array is indicated by a dark line, shown adjacent to imaging plane cross-section.
The volume consists of $256\times 256\times 128$ voxels, corresponding to dimensions of
$64\times64\times32$mm. 
}
\label{fig.vol3d}
\end{figure}

\begin{figure}[h]
\ifdefined\SHORTTEX
    \centerline{\includegraphics[width=1.0\columnwidth, trim=10px -5 16px 0, clip]{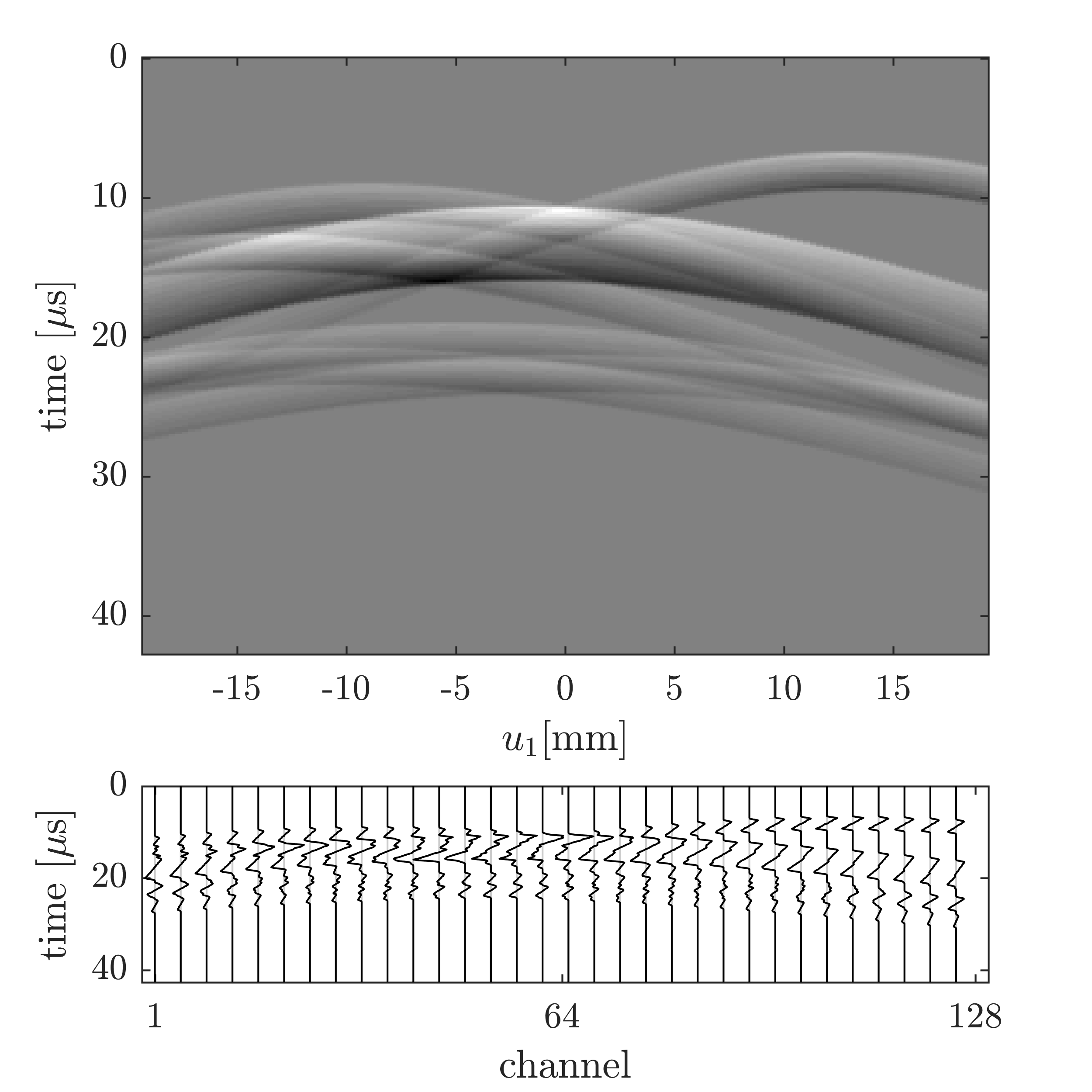}}
\else
    \centerline{\includegraphics[width=0.7\textwidth, trim=10px -5 16px 0, clip]{fig_sinogram1.png}}
\fi    
    \centering
    \vspace{-10px}
\caption[Simulated time-domain signals generated by opto-acoustic sources]{Simulated time-domain signals generated by opto-acoustic sources 
as configured in 
 Figure~\ref{fig.vol3d}. (top) Sinogram image showing signal intensity.
  (bottom) Line plot of the signals for selected channels. 
  }
\label{fig.sinogram1}
\end{figure}

\begin{table}[bh]
\caption[Computation time for simulating 3D volumes]{Computation time (milliseconds)
for
simulating 3D volumes based on
 Propositions~\ref{prop:s1}, \ref{prop:sep_conv} and \ref{prop:s4}.  
The 
grid size (number of voxels) 
was varied. 
All cases used 256 transducer channels with 1024 output samples per channel. 
The total execution time includes run-times for the GPU implementations of $\scrG$ and $\tilde\scrG$, plus computational overhead for memory transfers and CPU processing. 
}
\label{tab:sep_timings}
\centering
\begin{tabular}{ccccccc}
\hline
\textbf{Grid Size}  & 
\textbf{Method} & \raisebox{-.3ex}{$\bm\scrG$}  & \raisebox{-.3ex}{$\bm{\tilde\scrG}$} & \textbf{Overhead} & \textbf{Total}  
\\
\hline\hline
\centering
  $128\times 128\times 128$ & Prop.~\ref{prop:s1}			&	4	&	14	&   13  & 	31
\\$256\times 256\times 256$ & Prop.~\ref{prop:s1}			&	10	&	14	&	18	& 	42
\\$512\times 512\times 512$ & Prop.~\ref{prop:s1}			&	39	&	16	&	102	& 	157
\\$768\times 768\times 768$ & Prop.~\ref{prop:s1}			&	48	&	15	&	336	& 	399
\\\hline
  $128\times 128\times 128$ & Prop.~\ref{prop:sep_conv}		&	7	&	24	&	12	&	43
\\$256\times 256\times 256$ & Prop.~\ref{prop:sep_conv}		&	18	&	24	&	45	&	86
\\$512\times 512\times 512$ & Prop.~\ref{prop:sep_conv}		&	30	&	22	&	305	&	357
\\$768\times 768\times 768$ & Prop.~\ref{prop:sep_conv}		&	92	&	19	&  1001	&	1112
\\\hline
  $128\times 128\times 128$ & Prop.~\ref{prop:s4}			&	12	&	99	& 11	&	122
\\$256\times 256\times 256$ & Prop.~\ref{prop:s4}			&	25	&	96	& 20	&	141
\\$512\times 512\times 512$ & Prop.~\ref{prop:s4}			&	91	&	93	& 55	&	239
\\$768\times 768\times 768$ & Prop.~\ref{prop:s4}			&	208	&	92	& 365	&	665
\\\hline
\end{tabular}

\end{table}

To demonstrate the accuracy of the model, 
a volume consisting of a single point absorber was compared to time-domain impulse responses 
that were 
simulated using the acoustic software package Field II\cite{jensen1996field,jensen1992calculation}.  
In Figure~\ref{fig.exact_impulse}, the impulse response output for rectangular aperture using the exact separable method according to Proposition~\ref{prop:sep_conv} is shown. The output corresponds to signal produced by the point absorber, which is modeled as point source at a fixed location that is detected by different transducers in the linear-array. 
The output from the exact separable method closely matches that of Field II, except for differences that are attributed to a low-pass filtering effect due to the voxel spacing used in the simulation, where a linear interpolation kernel is used. The effect reduces the intensity of the peak impulses in time-domain, and the magnitude of the transfer-function in frequency-domain is reduced at high frequencies, as shown.   This occurs because Field II is capable of simulating a point source with a high bandwidth. However, our method is voxel-based. A voxel is a finite sized object which depends on the 3D grid sizing of the volume. Tri-linear interpolation filters 
smooth the  impulse response from each voxel, and this causes the observed time-domain low-pass effects.  However, the grid can be made arbitrarily fine to alleviate this if the number of voxels is increased, or if the volume dimensions are decreased, which can increase memory usage and/or processing time. 

In Figure~\ref{fig.approx_impulse}, the far-field approximation of Proposition~\ref{prop:s4} was used to generate time-domain output. As expected, it is seen that  the far-field approximation does not match the Field II output as closely as for the exact approach of Figure~\ref{fig.exact_impulse}, although there is still a close correspondence. Also, the timing of the rising and falling edges do match closely. Some features of the waveform shape 
and the approximate area of under the curves
are also captured by the approximation. 
Use of Proposition~\ref{prop:s4} is suitable 
when a 
balance between speed and accuracy is required. 

In Figure~\ref{fig:two_spheres_sinogram_uffc}, 
simulation 
with transducer aperture for
an ideal point detector is compared to a 4mm-wide line aperture. 
A digital phantom 
 consisting of two spherical absorbers was generated as shown in 
Figure~\ref{fig:two_spheres_3d_uffc}. 
This phantom was designed to produce {out-of-plane artifact} in 2D imaging, 
which results 
because
one sphere generates acoustic interference when the probe is positioned over the other. 
The spheres have 10mm diamter, a 16mm center depth, and  
 homogeneous optical absorption coefficient 
of 1 arb.u. 
The
 background was assumed to have negligible absorption. 
The simulated volume 
had
 $256\times 256 \times 128$ voxels, with grid spacing $\Delta x=0.5$mm/voxel. 
Both transducers 
modeled cosine directivity using the 
obliquity factor $\alpha(\bfx)$ according to Proposition~\ref{prop:s3}, 
with ideal electro-mechanical impulse response $h_\text{em}(t)=\delta(t)$. 
The spheres produced \enquote{N-shaped} waveforms, which was expected. 
The wide aperture provides increased directionality,
which filters out more signal from the out-of-plane sphere, as shown. 
Further analysis regarding image reconstruction using the dataset of Figure~\ref{fig:two_spheres_sinogram_uffc}
is described in our follow-up paper\cite{zalev2020convex}.

\begin{figure*}[t]
\begin{minipage}[t]{0.48\textwidth}
    \centerline{\includegraphics[width=1.0\textwidth, trim=70 0 70 0, clip]{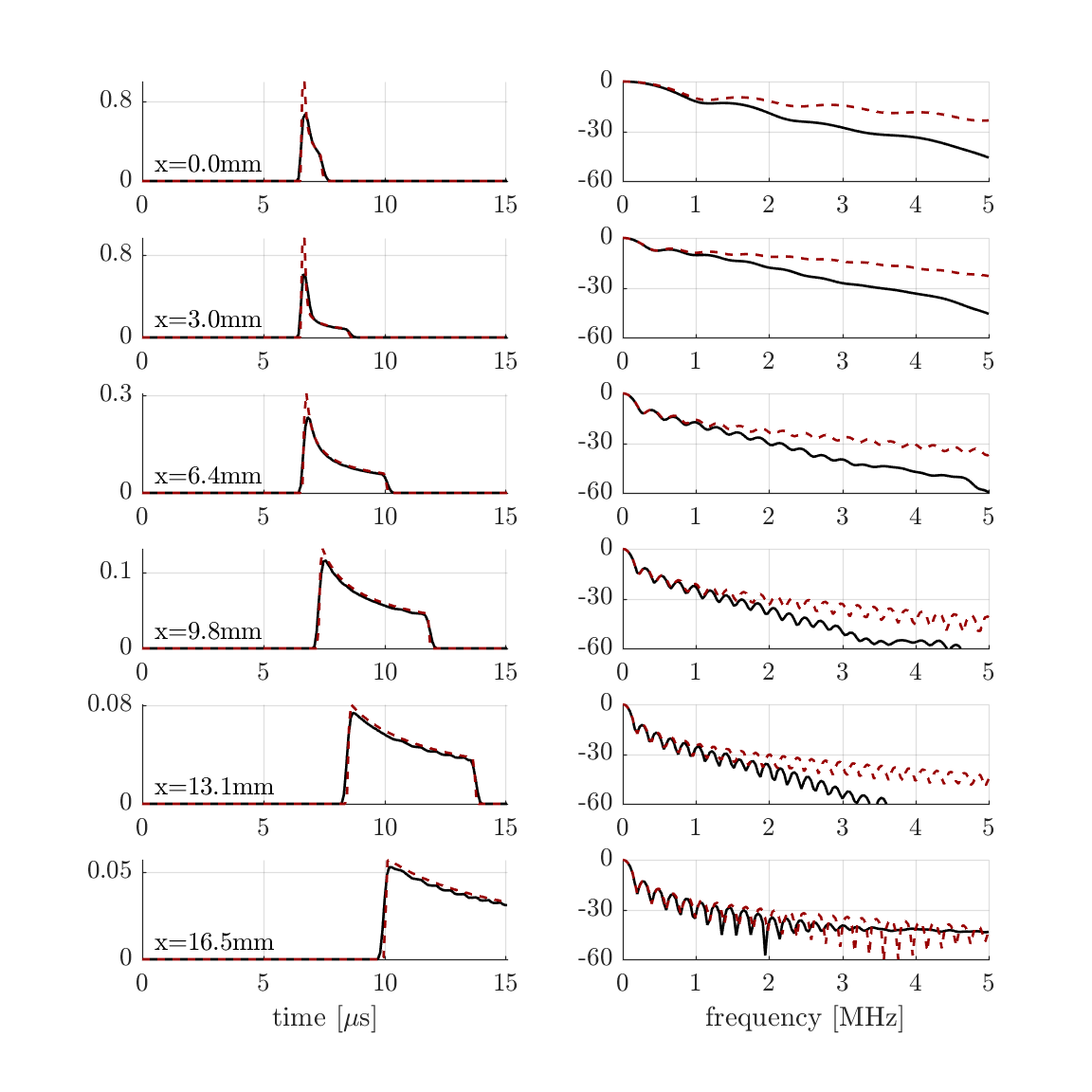}}
    \centering
    \vspace{-10px}
\caption[Signals from point source detected by transducer (separable approach)]{Signals from point source at (x,y,z)=(0,0,10)mm, detected by transducer of width $6\text{mm}\times6\text{mm}$. 
(solid line) computed using separable approach of Proposition 3.2. (dotted line) computed using Field II. 
(left) time-domain. (right) frequency domain. 
Due to the finite grid spacing of 3D voxels in the opto-acoustic source distribution, 
this results in a low-pass filtering effect of the simulated signals, which accounts for 
the differences seen here. 
 }
\label{fig.exact_impulse}
\end{minipage}%
\hfill\hfill%
\begin{minipage}[t]{0.48\textwidth}
    \centerline{\includegraphics[width=1.0\textwidth, trim=70 0 70 0, clip]{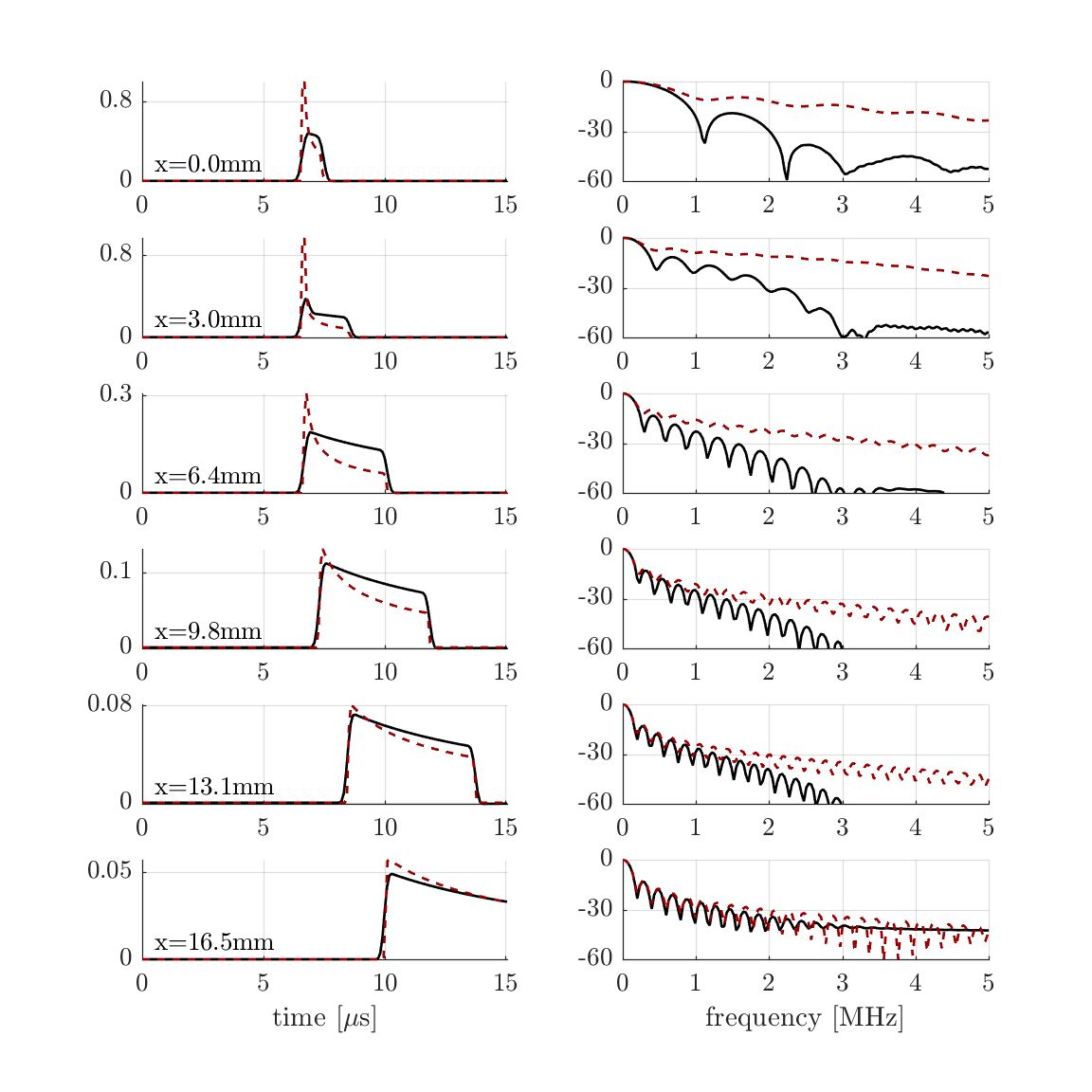}}
    \centering
    \vspace{-10px}
\caption[Signals from point source detected by transducer (approximation)]{
Signals from point source at (x,y,z)=(0,0,10)mm, detected by transducer of width $6\text{mm}\times6\text{mm}$. 
(solid line) computed using separable far-field approximation of Proposition 3.4. (dotted line) computed using Field II. 
(left) time-domain. (right) frequency domain. 
The far-field approximation of Proposition 3.4 results in increased computational performance
compared to Proposition 3.2, however the accuracy is reduced. 
}
\label{fig.approx_impulse}
\end{minipage}%
\hfill

\end{figure*}

\begin{figure*}[bt]
\centering
\subfloat[$0$mm width]{\label{fig:sinogram_0mm_uffc}\includegraphics[width=0.45\textwidth, trim=0px 0px 0px 30px, clip]{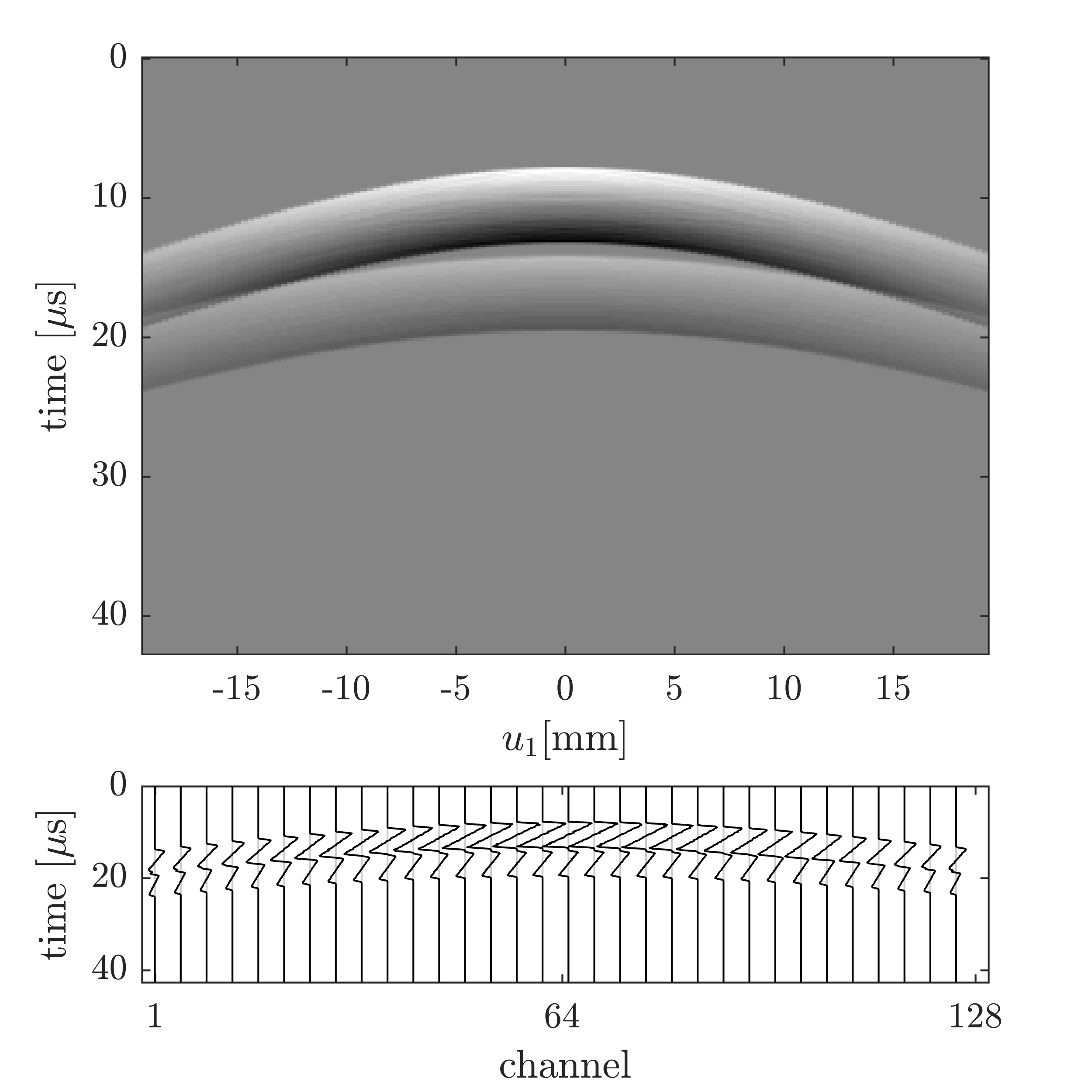}}%
\subfloat[$4$mm width]{\label{fig:sinogram_4mm_uffc}\includegraphics[width=0.45\textwidth,   trim=0px 0px 0px 30px, clip]{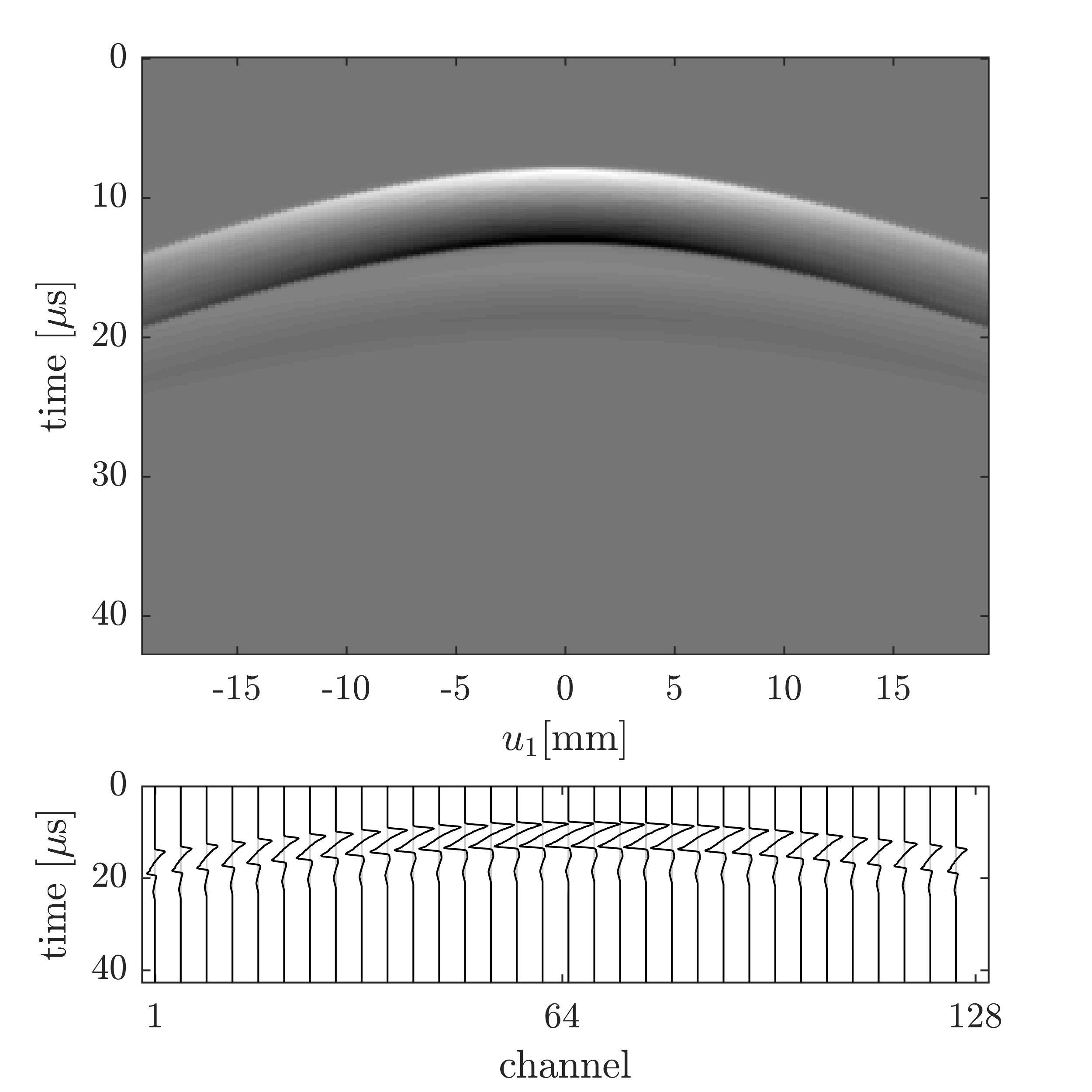}}
\caption[Simulated data for $0$mm and $4$mm width transducer]{\label{fig:two_spheres_sinogram_uffc} 
Simulated data for 
linear array. Time domain data is shown graphically (top) and plotted (bottom).  
The array is positioned as shown in Figure~\ref{fig:two_spheres_3d_uffc}.  
The transducer elements are modeled using (a) an ideal point detector and (b) a $4$mm wide line aperture. 
The time-domain plots (bottom) illustrate that each spherical absorber generates an N-shaped waveform, 
which arrives earlier for the in-plane absorber than the out-of-plane absorber. 
The out-of-plane waveform is weaker in (b) compared to (a) because the transducer aperture acts as a spatially-dependent filter.  
  }
\end{figure*}

\begin{figure}[b]
\centering
\ifdefined\SHORTTEX
\includegraphics[width=0.7\columnwidth, trim=0px 0px 0px 0px, clip]{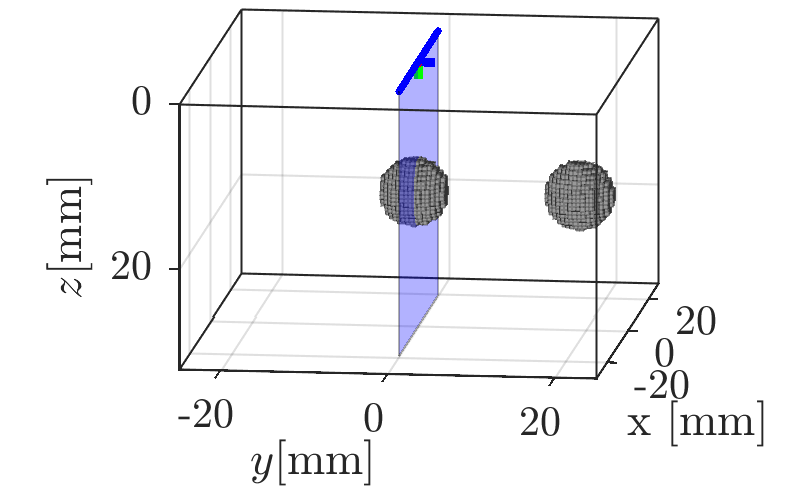}
\else
\includegraphics[width=0.9\textwidth, trim=0px 0px 0px 0px, clip]{two_spheres_3Da.png}
\fi
\caption[3D volume with spherical absorbers]{\label{fig:two_spheres_3d_uffc} 
3D volume with spherical absorbers. 
A linear array (blue line) 
and the imaging plane (light blue) are shown.
 }
\end{figure}

\section{Discussion}
\label{sec:discussion}

\begin{table}[h]
\caption{Computational complexity of simulation models.}
\label{tab:comp_complex}
\centering
\begin{tabular}{cc}
\hline
\textbf{Method} & \textbf{Computational Complexity} 
\\
\hline\hline
\centering
Non-separable & $\mathcal{O}\big( m_x m_y  n_x n_y n_z k_a\big)$ 
\\
\centering
Sep. Exact (Prop.~\ref{prop:s1}~\&~\ref{prop:sep_conv}) &  
$\mathcal{O}\big( m_y m_x  (n_y n_z+ n_x n_z)  k_b\big)$ 
\\
\centering
Sep. Approx. (Prop.~\ref{prop:s4}) & 
$\mathcal{O}\big(m_y  m_x (n_y n_z+ n_x n_z + k_c) \big)$ 
\\
\centering
Mixed-domain  & 
$\mathcal{O}\big(m_y n_z (m_x n_y  k_d +   n_x \log(n_x n_z))\big)$\\
\hline
\end{tabular}

\vspace{5px}

\caption{Symbol names for complexity analysis}
\label{tab:comp_complex_notation}
\centering
\begin{tabular}{cc}
\hline
{\centering \textbf{Symbol}} &  
{\centering \textbf{Description}} \\
\hline\hline
$n_x$, $n_y$, $n_z$& number of voxels in $x$-, $y$- and $z$-axis \\
$m_x$& number of transducers in linear-array \\
$m_{x'}$& number of intermediate positions   \\
$m_y$& number of frames (or transducer rows)  \\
$k_a$, $k_b$, $k_d$ & spatial samples per output sample\\
$k_c$ & temporal samples per output sample \\
\hline
\end{tabular}
\vspace{10px}
{
}
\end{table}

As mentioned earlier, it may be less efficient to perform spatial convolutions per the separable method of Proposition~\ref{prop:sep_conv}, compared to the separable far-field approximation described in Proposition~\ref{prop:s4}. 
This is because in modern programmable computing architectures, the time required to perform arithmetic operations may be dwarfed by the time required for non-cached memory accesses. This means that copying data from one memory bank to another, which requires little or no arithmetic processing, could take nearly as much processing time as a well implemented but extensive mathematical transformation performed on the same data.
 For computing a convolution with an aperture function, by including a step of forming a temporary intermediate output, or by having to fetch additional samples from memory, this can be less efficient than grouping a few extra arithmetic operations within the separable cascade.
When speed is more critical than accuracy, a far-field approximation such as the one proposed in  Proposition~\ref{prop:s4} can improve throughput. 
  However, if higher accuracy is required, then Proposition~\ref{prop:sep_conv}
may be preferred over Proposition~\ref{prop:s4}. Nevertheless, even when Proposition~\ref{prop:sep_conv} is used, the computation is still an order-of-magnitude reduced compared to a non-separable implementation. 

An analysis of the computational complexity for each approach is summarized in Table~\ref{tab:comp_complex}. 
In Table~\ref{tab:comp_complex_notation}, a description of the variables used in this analysis is provided.
In the 3D grid, the number of voxels in the $x$, $y$ and $z$ directions are $n_x$, $n_y$ and $n_z$. 
The total number of voxels is $N=n_x \times n_y \times n_z$.
The number of frames that are acquired by the linear array is $m_y$
(we assume that there is one frame acquired per laser pulse). 
Since each frame uses all transducer elements in the linear-array to record data,  
 $m_y$ would also correspond to the number of rows in a  matrix-array transducer, if each frame is treated as a row.  
The number of transducers in the linear array is $m_x$. To group these, we have $M=m_x \times m_y$.
In future work\cite{zalev2020convex}, we describe the case where multiple frames and multiple optical wavelengths are used. 

For a non-separable approach, the complexity is 
$\mathcal{O}( m_x m_y  n_x n_y n_z k_a)=\mathcal{O}( M N k_a)$. 
The parameter $k_a$ is the ratio of spatial samples (voxels) that must be processed per time-domain sample per transducer channel. 
In other words, each voxel maps onto $k_a$ time-domain samples per channel. 
Various implementation design decisions effect $k_a$. 

For a separable approach, the computational complexity is around $\mathcal{O}( M N^{\frac 2 3} k)$, where $k$ represents 
an implementation dependent parameter (either $k_a$, $k_b$, $k_c$ or $k_d$, as described below). 
The complexity is therefore reduced by an order-of-magnitude compared to a non-separable approach. 
For example, consider a volume of $100\times 100\times100$ voxels. Here, $N=100^3$ and $N^{\frac 2 3}=100^2$. This means that for the separable case there are 100 times fewer operations than with a non-separable approach. 
In this case, a 100 times speed-up would easily be achievable in practice, and the speed-up would be higher for larger volumes.

The parameter $k_a$ is  proportionate to the number of samples required to model the transducer aperture geometry.  
It is also proportionate to the ratio $\frac{f_s}{f_c}$, 
where $f_c=\frac{c_0}{\Delta x}$ is the critical sampling rate (described above), 
$f_s=\frac{n_s}{t_\text{max}}$ is the sampling rate used in simulation, 
$n_s$ is number of time-domain samples, and $t_\text{max}$ is the sampling duration. 
It is assumed that sampling occurs at the critical rate $f_s=f_c$. 
For a point aperture, $k_a$ is assumed to equal 1.

For the separable approach, during the first application of the operator $\scrG$, which is an intermediate step, 
the number of output channels is not 
the number of transducers $m_x$ (this matches after the second application of $\scrG$). The number of output positions for the intermediate step is called $m_{x'}$. It is recommended that $m_{x'}$ should be sufficiently greater than $m_x$ to avoid boundary problems, so that the first and final transducers have enough adjacent data for accurate simulation. Using $m_{x'}=3 m_{x}$ will supply additional padding equal to the length of the array. 
For Propositions~\ref{prop:s1}~or~\ref{prop:sep_conv},  the complexity is $\mathcal{O}\big( m_y  (m_{x'}n_x n_z+ m_x n_y n_z)  k_b\big)$. 
Generally,  $m_x$ can be assumed to be a constant multiple of $ m_{x'}$, so they are the same for complexity purposes.
The parameter $k_b$ is similar to $k_a$ although its implementation may be different. 
For a point aperture, 
Proposition~\ref{prop:sep_conv}
reduces to 
Proposition~\ref{prop:s1} and $k_b=1$. 
In our implementation of Proposition~\ref{prop:sep_conv}, 
the convolution $\scrG \circledast f$ is implemented separate from the GPU computation of $\scrG$, which 
adds computational overhead with complexity $\mathcal{O}\big(n_x n_y n_z k_b m_y\big)$ that becomes dominant in large grid sizes. 
This is one reason why Proposition~\ref{prop:s4} performs better in Table~\ref{tab:sep_timings} for large grid sizes.

When the far-field approximation of Proposition~\ref{prop:s4} is used, 
spatial-domain filtering operations are moved to time-domain operations. 
In this case, $k_b$ would have a constant value of 2, and a new parameter $k_c$ is introduced corresponding to additional time-domain sampling operations. Thus, if $k_c$ is only proportional to the number of output channels, the complexity is changed from 
$\mathcal{O}\big( m_y m_x  (n_x n_z+ n_y n_z)  k_b\big)$ to 
$\mathcal{O}\big(m_y  m_x (n_x n_z+ n_y n_z + k_c) \big)$. This is another reason why  
Proposition~\ref{prop:s4} can be made faster than Proposition~\ref{prop:sep_conv}.  
Both reduce to approximately the same complexity for Proposition~\ref{prop:s1}, where $k_b=k_c=1$. 
If it is desired to include the obliquity factor according to Proposition~\ref{prop:s3}, 
since this consists of time-domain operations, it 
corresponds to a mixture of the $k_b$ and $k_c$ parameters, depending on implementation, and the complexity analysis is similar.

Table~\ref{tab:sep_timings} illustrates run-times measured using our implementation. 
In Table~\ref{tab:sep_timings}, 
the run-time for $\scrG$ varies in proportion to grid size because the intermediate data is critically sampled relative to the 3D volume.
However, 
in Table~\ref{tab:sep_timings}, 
the run-time of $\tilde \scrG$ does not change significantly (within each method) with grid size because 
all tests used a constant $n_s=1024$. 
If $n_s$ is increased with grid size, the run-time for $\tilde \scrG$ scales in a similar manner to $\scrG$. 
For large grid sizes, the run-times for $\tilde \scrG$ and $\scrG$ are significantly faster compared to the computational overhead. 
In applications such as image reconstruction, where $\tilde \scrG$ and $\scrG$ (and their adjoints) must be repeated multiple times, 
the the computational overhead becomes less significant because portions of it are performed only once.

Separable computation can also be performed by separable frequency-domain approaches\cite{jakubowicz1989two}. In the frequency-domain, this is computationally efficient when transducers are aligned on a grid. This is because the Fast Fourier Transform (FFT) uses spatial sampling at regular intervals. The computational complexity of performing a 2D-FFT for an image of size $n_y \times n_z$ is $\mathcal O (n_y n_z \log n_y n_z)$. 
To perform separable simulation in the frequency-domain, this involves using a pair of two-dimensional cascaded operators (analogous to the cascaded time-domain operators). 
In each cascaded 2D frequency-domain operation, the space of 2D Fourier data is warped using a 
non-linear transformation. This can be implemented as interpolation and re-sampling operations, so this does not increase computational complexity, aside from 
determining 
coordinate indices of a non-linear function. 
Convolution with an aperture function, 
can  be performed in the frequency-domain 
 using ordinary multiplication. An inverse 2D-FFT applied to the warped data generates the final output for $\scrG$. 

The frequency-domain output corresponds to the same 
 2D data 
  expected by 
   $\scrG$ (see Figure~\ref{fig.probe_geom_pll}b). 
Consequently, a separable time-domain factor can be replaced with its frequency-domain equivalent. 
The advantage to this replacement is that the frequency-domain version is able to compute data for all transducer positions simultaneously, using the $\mathcal O(n\log n) $ FFT complexity. 

However, when a linear-array undergoes motion along an arbitrary trajectory, the
acquisitions (each coinciding with a laser pulse) occur when the probe has non-uniformly spaced position and orientation. 
 This means that using a frequency-domain version of $\scrG$ 
 loses its efficiency, since the planes in Figure~\ref{fig.probe_geom_pll}a would not line up, unless the probe moves uniformly. 
The lowest computational complexity would arise when a mixed-domain approach is used, 
where the first application of $\scrG$ uses the time-domain approach, and the second application of $\scrG$ uses the frequency-domain. The computational complexity for this is   
$\mathcal{O}\big(m_y n_z (m_x n_y  k_d +  k_e n_x \log(n_x n_z))\big)$, where $k_d$ and $k_e$ are implementation dependent parameters. 
Moreover, convolution of the transducer aperture function can  be applied efficiently in the frequency domain using multiplication. 
In this paper, only the mathematical details  for time-domain based operators was discussed. 
Details for comparing the output of time-domain to frequency-domain (and mixed-domain) simulations is the subject of future work. 

Additional considerations can permit modeling of an acoustic lens or cylindrical transducers. 
For example, the contour path illustrated in Figure~\ref{fig.arcpath} can be modified to include temporal delays that approximate such effects. In this manner, the geometry for a circularly cylindrical transducer can be modeled using separable operators similar to the approach of Proposition~\ref{prop:sep_conv}, which is a subject of further study.

\section{Conclusion}
\label{sec:concl}

A fast approach 
to simulate
  opto-acoustic signal generation
in a linear array of 
 rectangular transducer elements was described. 
Volume integration using a Green's function solution of the acoustic wave equation 
was
  shown to be separable 
 into a pair of operators that act along perpendicular directions,  
which permits fast algorithms for computing opto-acoustic response. 
The computational complexity for time-domain simulation of a volume with $N=n_x \times n_y \times n_z$ voxels, as recorded by $M=m_x \times m_y$ transducers, is reduced from
 $\mathcal{O}(MNk^2)$, in the case where a non-separable approach is used, to $\mathcal{O}(m_x m_y k(n_x n_z + n_y n_z))\approx\mathcal{O}(MN^\frac{2}{3}k)$, where $k$  depends on transducer geometry. 
 A custom GPU-based implementation was used to demonstrate this technique.

\begin{subappendices}
\ifdefined\SHORTTEX
\else
\bgroup
\fi

\section{Appendix}
\subsection{Proof of Proposition 3.1}
\label{sec:proof1}

\begin{proof}
First, we examine the Green's function 
from Equation~\eqref{eq:g_defn}, which is
\begin{align*}
&g(\bfx, t) 
=
\frac{1}{4\pi \norm\bfx}
\delta\!\left(
t-
\norm{
\bfx}
\right)
=
\frac{1}{4\pi t}
\delta\!\left(
t-
\norm{
\bfx}
\right)
.
\end{align*}
It contains  $\delta\!\left(
t-
\norm{
\bfx}
\right)$, which can be factored according to the sifting property of the delta function%
\footnote{
The sifting property is
$\int_{-\infty}^{\infty} \delta(x-y) f(x) \,\d x = f(y)$.
}.  
Since $\norm\bfx^2 = x_1^2+(x_2^2+x_3^2)$, 
it follows that
\begin{align}
\delta\!\left(
t-
\norm{
\bfx}
\right) 
=
\int_{-\infty}^{\infty} 
\delta
\!\left(t-
r
\right)
\delta
\!\left(\tau-
\rho
\right)
\d \tau,
\end{align}
where 
$r = \sqrt{x_1^2+\tau^2}$
and
$\rho = \sqrt{x_2^2+x_3^2}$.
Therefore, 
\begin{align}
h(\bfx, t) &=
\frac{\partial}{\partial t}g(\bfx,t) 
\nonumber
\\&=
\frac{1}{4\pi } \frac{\partial}{\partial t} 
\left(
\frac{1}{t}\!
\int 
\delta
\!\left(t-
r
\right)
\delta
\!\left(\tau-
\rho
\right)
\d \tau
\right).
\label{eq:prop1_hconv}
\end{align}

From Equation~\eqref{eq:p_eq_h_conv_psi},
the pressure $p(\bfx, t)$ measured from an acoustic source distribution $\psi(\bfx)$ is
\begin{align*}
\scrH\{\psi\}(\bfx, t) &=
h(\bfx, t) \conv{\bfx} \psi(\bfx)
\\&
= \int h(\bfx',t) \psi(\bfx - \bfx') \d \bfx'.
\end{align*}
By substituting \eqref{eq:prop1_hconv}, 
this is equal to 
\begin{align}
\frac{1}{4\pi} \frac{\partial}{\partial t} 
\iint
\frac{1}{t}\,
\delta
\!\left(t-
r'
\right)
\delta
\!\left(\tau-
\rho'
\right)
\psi(\bfx -\bfx') 
\,
\d \tau
\d\bfx',
\label{eq:split_delta}
\end{align}
where $r'=\sqrt{x_1'^2+\tau^2}$ and $\rho'=\sqrt{x_2'^2+x_3'^2}$.
Since $\d\bfx'=\d x_1' \d x_2' \d x_3'$, 
by replacing $\d{x_2'}\d{x_3'}$ with 
$\rho' \d{\rho'} \d{\theta'}$, where $\theta'\in[0,2\pi]$, 
the inner integral can be written in polar coordinates. Equation~\eqref{eq:split_delta} becomes
\begin{align*}
\frac{\partial}{\partial t} 
{
\iint \!\!
\frac{
\delta
\!\left(t-
r'
\right)
}
{4\pi t}
\!\!
\underbrace{
\int_0^{2\pi}\!\!\!\!\int_0^\infty\!\!
\delta
\!\left(\tau-
\rho'
\right)
\psi(\bfx -\bfx') 
\rho' 
\,
\d \rho' \d \theta'
}_{
\sigma(x_1-x_1',x_2,\tau)
}
\d \tau
\d x_1'
}
.
\end{align*}
Here, the inner integral is grouped as a function that depends on $\tau$, $\bfx$ and $\bfx'$.
In the inner integral, $\rho'$ can be eliminated (and replaced with $\tau$) by applying the sifting property to 
$\delta(\tau - \rho')$.  
Since $x_2'=\rho'\cos\theta'$ and $x_3'=\rho'\sin\theta'$, 
 the inner integral becomes
\begin{align}
\label{eq:sigma_def1}
\tau
\!
\int_0^{2\pi}
\!
\psi\!\left(
\bfx-
\begin{bmatrix}
0\\
\tau \cos\theta'\\
\tau \sin\theta'
\end{bmatrix}
-
\begin{bmatrix}
x_1'\\
0\\
0
\end{bmatrix}
\right)
\,
 \d \theta'.
\end{align}
To reduce the number of variables that this expression depends on, the coordinate
$x_3$ is set to $0$, so the transducer is constrained to the $x_1x_2$-plane. 
We define $\sigma$ as
\begin{align}
\label{eq:prop1_sigma}
\sigma(x_1,x_2, \tau)=
\tau
\!
\int_0^{2\pi}
\!
\psi\!\left(
\begin{bmatrix}
x_1\\
x_2\\
0
\end{bmatrix}
-
\begin{bmatrix}
0\\
\tau \cos\theta\\
\tau \sin\theta
\end{bmatrix}
\right)
\,
 \d \theta.
\end{align}

The overall expression (with $x_3 = 0$) becomes
\begin{align*}
\frac{1}{4\pi} 
\frac{\partial}{\partial t} 
{
\iint \!
\frac{
\delta
\!\left(t-
r'
\right)
}{ t} 
\sigma(x_1-x_1',x_2, \tau)
\d \tau
\d x_1'
}
.
\end{align*}
Since $r'=\sqrt{x_1'^2+\tau^2}$, this can be written in polar coordinates 
where 
$x_1'=r'\cos\theta$,
$\tau=r'\sin\theta$, 
and $\theta \in [0,2\pi]$. 
By replacing
$\d \tau \d x_1'$ with $r'\d{r'}\d \theta$, and 
using the sifting property on $\delta(t-r')$ to eliminate $r'$, 
 the  expression becomes
\begin{align}
&\frac{1}{4\pi} \frac{\partial}{\partial t} 
\frac{1}{t}
{
\iint \!
{
\delta
\!\left(t-
r'
\right)
}
\sigma(x_1-r' \cos \theta,x_2,r' \sin \theta)
r'
\d r'
\d \theta
}
\nonumber
\\
&=
\frac{1}{4\pi } \frac{\partial}{\partial t} 
\frac{1}{t}
\underbrace
{
t \!
\int \!
\!
\sigma(x_1-t \cos \theta,x_2,t \sin \theta)
\d \theta
}_{\tilde s(x_1, x_2,t)}
\nonumber
\\
&=
\frac{1}{4\pi } \frac{\partial}{\partial t} 
\frac{
\tilde s(x_1, x_2,t)}
{t},
\end{align}
where
\begin{align}
\label{eq:s_tilde_def1}
\tilde s(x_1, x_2, t) =
t
\!
\int
\!
\sigma\!\left(
\begin{bmatrix}
x_1\\
x_2\\
0
\end{bmatrix}
-
\begin{bmatrix}
t \cos\theta\\
0\\
t \sin\theta
\end{bmatrix}
\right)
\d \theta.
\end{align}
Here, the three arguments applied to $\sigma$ are arranged as a column vector to highlight the similarity with 
\eqref{eq:prop1_sigma}. 

Let the operator $\scrG^{A,\bfb}$ be defined as 
\begin{align}
\scrG^{A,\bfb}\{\psi\}(u_1, u_2, t)
= \scrG\{\scrT^{-1}_{A,\bfb}\{\psi\}\}(u_1, u_2, t), 
\end{align}
where 
\begin{align}
\label{eq:scrG_def}
\scrG\{\psi\}(u_1,u_2, t)
= t \! \int_{0}^{2\pi} \! \psi \left( 
\begin{bmatrix}
u_1\\u_2-t \cos \theta \\ -t \sin \theta
\end{bmatrix}
 \right)
 \d \theta,  	
\end{align}
and 
\begin{align*}
\scrT^{-1}_{A,\bfb}\{\psi\}(\bfu) = \psi(A \bfu + \bfb).
\end{align*}

Then $\scrG^{A,\bfb}\{\psi\}(u_1, u_2, t)$ is equal to
\begin{align}
t \! \int_{0}^{2\pi} \! \psi \left( 
A
\begin{bmatrix}
u_1\\u_2-t \cos \theta \\ -t \sin \theta
\end{bmatrix}
+
\bfb
 \right)
 \d \theta.  	
\end{align}

Using Equation~\eqref{eq:prop1_sigma}, 
 $\sigma(u_1, u_2, \tau)$ can be written as
\begin{align*}
\scrG^{A, \bfb}\{\psi\}(u_1, u_2, \tau), 
\end{align*}
when $A=I$ and $\bfb=(0, 0, 0)$. 
Also, from \eqref{eq:s_tilde_def1}, $\tilde s(u_1, u_2, t)$ can be written as
\begin{align*}
\scrG^{R, \bfzero}\{\sigma\}(u_2, -u_1, t), 
\end{align*}
with
the order of arguments $u_1$ and $u_2$  reversed, 
and, using \eqref{eq:R3_yaw}), 
$$R=R_3(90^\circ)=\begin{bmatrix}
0&-1&0\\
1&0&0\\
0&0&1
\end{bmatrix} . 
$$
In this equation,
the rotation matrix $R_3(90^\circ)$ switches the ordering of the 
first two elements in a vector that is applied to it, with a sign change. 
In particular, 
\begin{align*}
R_3(90^\circ)\begin{bmatrix}
u_2\\-u_1-t\cos\theta\\ -t \sin\theta
\end{bmatrix}
=
\begin{bmatrix}
u_1+t\cos\theta\\u_2\\-t \sin\theta
\end{bmatrix}.
\end{align*}
Here, due to symmetry
of integration over $\theta$, the sign change has no impact.

Therefore, the operator $\scrH^{A,\bfb}\{\psi\}(u_1,u_2, t)$ can be written
\begin{align*}
\frac{1}{4\pi } \frac{\partial}{\partial t} 
\Big(
\frac{1}{t}\,
\scrG^{R,\bfzero} \{  
\scrG^{A,\bfb} \{  
\psi
\}\}(u_2,-u_1,t)
\Big).
\end{align*}
This can be written more compactly as
\begin{align*}
\scrH^{A,\bfb} 
&=
\tilde\scrG
\circ 
\scrG^{A,\bfb}
,
\end{align*}
where
\begin{align*}
\tilde\scrG \{\psi\}(u_1, u_2, t) 
=
\frac{1}{4\pi } \frac{\partial}{\partial t}
\!
\left(\frac{1}{t}\,
\scrG^{R,\bfzero}
\{\psi\}(u_2, -u_1, t)
\right)
.
\end{align*}
This completes the proof.
\end{proof}

\subsection{Proof of Proposition 3.2}
\label{sec:proof2}

\begin{proof}
This proof follows from expanding the convolution integral 
from Proposition~\ref{prop:s1}, 
when multiplicatively separable aperture $f$ is included.
We illustrate this with an expanded version of
Equation \eqref{eq:split_delta}  
when written as a convolution. 

First we observe that for an arbitrary aperture function $f(\bfu)$, 
the operator $\scrH_f^{A,\bfb}\{\psi\}(\bfu,t)$ can be written as a convolution 
\begin{align*}
\scrH_f^{A,\bfb}\{\psi\}(\bfu) &= f(\bfu)\conv{\bfu} 
\left(
\scrH^{A,\bfb}\{\psi\}(\bfu,t) \right)
\\&=
\scrH\{f \conv{} \scrT_{A,\bfb}^{-1}\{\psi\} \}(\bfu,t) . 
\end{align*}
This can be seen from Equation~\eqref{eq:scrH_bfx_t_def}, 
where by definition
\begin{align*}
\scrH_f\{\psi\}(\bfx) &=
\underbrace{
h(\bfx,t)\conv{\bfx}f(\bfx)}_{h_f(\bfx,t)}\conv{\bfx}\,\psi(\bfx)\\
&= h_f(\bfx,t)\conv{\bfx}\psi(\bfx).
\end{align*}
In the context of separability, we already know from Proposition~\ref{prop:s1} that 
the impulse response $\scrH$ is separable. 
Therefore, if we define 
\begin{align*}
\psi_f^{A,\bfb}(\bfx) = \psi^{A,\bfb}(\bfx) \conv{\bfx} f(\bfx), 
\end{align*}
then $\scrH\{\psi_f^{A,\bfb}\}$ can also be computed by separable cascade per Proposition~\ref{prop:s1}.

We now need to show that $f_1$ and $f_2$ can be paired with each of the cascaded operators. 
Since $f(\bfu) = f_1(u_1) f_2(u_2)\delta(u_3)$, 
it follows that
\begin{align}
\label{eq:psi_AB_conv_bfu_f}
\psi^{A,\bfb}(\bfu) \conv{\bfu} f(\bfu)
=
\psi^{A,\bfb}(\bfu) \conv{u_1} f_1(u_1) \conv{u_2} f_2(u_2).  
\end{align}
Here, we have used the relationship  $f_1(u_1)f_2(u_2)=f_1(u_1)\delta(u_2)\conv{u_1, u_2}f_2(u_2)\delta(u_1)$.

The operator
 $\scrG$ can be written as a 2D 
spatial convolution
with a cylindrical kernel (or a 3D spatial convolution with a ring shaped kernel). 
This can be seen 
by manipulating equation 
\eqref{eq:scrG_def}, first
by integrating with $\delta(t-\rho)$ to add the variable $\rho$ (which had been removed earlier by sifting), 
and then substituting $\d{u_2'}\d{u_3'}=\rho\d\theta \d\rho$. 
Thus 
\begin{align}
&\scrG\{\psi\}(u_1,u_2, t)
 \nonumber
= t \! \int_{0}^{2\pi} \! \psi \left( 
\begin{bmatrix}
u_1\\u_2-t \cos \theta \\ -t \sin \theta
\end{bmatrix}
 \right)
 \d \theta 
  \nonumber 	
\\&= 
 \rho \! \int_{0}^{\infty} \!\! \int_{0}^{2\pi} \! 
\delta(t-\rho)\, 
 \psi \left( 
\begin{bmatrix}
u_1\\u_2-\rho \cos \theta \\ -\rho \sin \theta
\end{bmatrix}
 \right)
 \d \theta \d \rho 	
  \nonumber
 \\&=
 \iint 
 \delta\!\left(t-\sqrt{u_2'^2+u_3'^2}\right)
 \psi\! \left( 
\begin{bmatrix}
u_1\\u_2-u_2' \\ 0 -u_3'
\end{bmatrix}
 \right) 
 \d{u_2'}\d{u_3'}
 \nonumber
\\&=
\left[
 \delta\!\left(t-\sqrt{u_2^2+u_3^2}\right)
 \conv{u_2, u_3}
 \psi(u_1, u_2, u_3) \right]_{u_3=0} 
 \nonumber
\\&=
\label{eq:G_eq_delta_conv}
\left[
\left(
 \delta(u_1)\delta\!\left(t-\norm\bfu\right)\right)
 \conv{\bfu}
 \psi(\bfu) \right]_{u_3=0}
.
\end{align}

With similar form to equation \eqref{eq:split_delta}, 
we can write $ \scrG^{R,\bfzero} \big\{ \scrG \{ \psi_f^{A,\bfb}(\bfu) \} \big\}(u_1, u_2, t)$ as
\begin{align*}
\left[
\delta\!\left(t-\norm{\begin{bmatrix}
u_1\\ \tau
\end{bmatrix}}\right)
\conv{u_1, \tau}
\delta\!\left(\tau-\norm{\begin{bmatrix}
u_2\\u_3
\end{bmatrix}}\right)
\conv{u_2, u_3}
\psi_f^{A,\bfb}(\bfu)
\right]_{u_3=0}. 
\end{align*}
Rearranging, and using Equation~\eqref{eq:psi_AB_conv_bfu_f}, this is equal to
\begin{align*}
&\left[
\left(
f_1(u_1)\conv{u_1}
\delta\!\left(t-\norm{\begin{bmatrix}
u_1\\ \tau
\end{bmatrix}}\right)
\right)
\conv{u_1, \tau}
\right.
\ifdefined\SHORTTEX
\eqbreak
\else
\fi
\left.
\left(
f_2(u_2)\conv{u_2}
\delta\!\left(\tau-\norm{\begin{bmatrix}
u_2\\u_3
\end{bmatrix}}\right)
\conv{u_2, u_3}
\psi^{A,\bfb}(\bfu)
\right)
\right]_{u_3=0}. 
\end{align*}
Thus, 
\begin{align*}
\scrG^{R,\bfzero} \big\{ \scrG \{ \psi_f^{A,\bfb}(\bfu) \} \big\}
=
\scrG_{f_1}^{R,\bfzero} \big\{ \scrG_{f_2} \{ \psi^{A,\bfb}(\bfu) \} \big\}. 
\end{align*}
Since $\tilde\scrG= \frac{1}{4\pi}\frac{\partial}{\partial t}\frac{1}{t} \scrG^{R,\bfzero}$,
which can be performed independently from the convolution, 
this yields $\tilde\scrH_f^{A, \bfb}$ and
the proof is complete.
\end{proof}

\subsection{Proof of Proposition 3.3}
\label{sec:proof3}

\begin{proof}
This proof breaks down into showing 
the following three cases individually, and their combination yields the full 
$\scrH_f^{A,\bfb}$: 
\begin{itemize}
\item[i.)]
$\bar\scrH^{A,\bfb}\{\psi\}(u_1, u_2, t)=[\bar h(\bfu)\conv{\bfu}\psi(T_{A,\bfb}(\bfu)]_{u_3=0}$
\item[ii.)]
$\bar\scrH_f\{\psi\}(u_1, u_2, t)=\bar\scrH\{\psi\}(u_1, u_2, t)\conv{u_1,u_2}f(u_1, u_2)$
\item[iii.)]
$
\bar\scrH\{ \psi\}(u_1,u_2,t)
=\frac{1}{t} \scrH\{\tilde \psi\}(u_1,u_2,t)
$
\end{itemize}

Case i.) This follows from the proofs of Sections~\ref{sec:separable_A} and
\ref{sec:separable_B}, where the reasoning is unchanged by the presence of the obliquity factor. 
Thus the separable cascade remains unchanged by the presence of rotation and translation. 

Case ii.) This follows from the properties of convolution, when the aperture function $f$ is included. 
Since the aperture is flat, the 3D convolution reduces to a 2D convolution. 
Since $\bar\scrH_f\{\psi\}(u_1, u_2, t)=\bar h_f(\bfu, t)\conv{\bfu}\psi(\bfu)$, we only need to confirm that $\bar h_f(\bfu,t)$
remains separable when aperture $f$ is present. 
Accordingly, 
\begin{align*}
\bar h_f(\bfu,t)&=
\bar h(\bfu, t)\conv{\bfu}f(\bfu) 
\\&=
\bar h(\bfu, t)\conv{\bfu}f(u_1, u_2)\delta(u_3) 
\\&=
\bar h(\bfu, t)\conv{u_1,u_2}f(u_1, u_2).
\end{align*}
This implies that when the obliquity factor is included, proper convolution with the aperture function is preserved.

Case iii.) In this case, the simplification also arises because the transducer is flat. 
For a general rotated frame, $\hat\bfn$ can be any unit vector. The operator  $\bar\scrH$, as defined in  Equation~\eqref{eq:bar_h_f_defn},  can be written of a sum of three components, 
resulting from the dot product in $\bfx \cdot \hat\bfn$.  Thus,
\begin{align*}
\bar\scrH \{\psi\}(\bfx, t) &=
\big[h(\bfx, t) \alpha(\bfx)\big] \conv{\bfx} \psi(\bfx)
\\&=
\left[h(\bfx,t) \frac{\bfx \cdot \hat\bfn}{\norm \bfx} \right] \conv{\bfx} \psi(\bfx) 
\\&=
\sum_i  \left(
\left[h(\bfx,t) \alpha_i(\bfx) \right] \conv{\bfx} \psi(\bfx) \right),
\end{align*}
where
$\alpha_i(\bfx)=\frac{x_i n_i}{\norm\bfx}$.
When $\hat\bfn = (0,0,1)$, 
which is the case when signals are measured in the local frame $\Fp$,
then the summation reduces to a single term, and $\alpha(\bfx)=\frac{x_3}{\norm\bfx}$.  

This is always the case for the local coordinate system $\bfu$, 
where the transducer is in the plane $u_3=0$. The aperture function is $f(\bfu)=f(u_1, u_2)\delta(u_3)$. When $f(\bfu)$ is convolved with $\bar h(\bfu)$, the component
$\hat\bfe_3\cdot (\bfu-\bfu')$ will always be equal to $u_3$ due to the factor $\delta(u_3)$.

From the convolution, we have 
\begin{align*}
\bar\scrH\{\psi\}(u_1, u_2, t) &=
\left[
\bar h(\bfu, t)\conv{\bfu}\psi(\bfu)
\right]_{u_3=0}
\\&=
\left[\left(
h(\bfu, t)\frac{u_3}{\norm\bfu}\right)
\conv{\bfu}\psi(\bfu)
\right]_{u_3=0}
\\&=
\left[
\int\!
h(\bfu', t)
\frac{u_3'}{\norm{\bfu'}}
\psi(\bfu-\bfu')
\,\d\bfu
\right]_{u_3=0}
\end{align*}
which (because $u_3=0$) is equal to
\begin{align*}
\int\!
\frac{h(\bfu', t)}{\norm{\bfu'}}
\underbrace{
u_3'
\,\psi\!\left(
\begin{bmatrix}
u_1-u_1'\\
u_2-u_2'\\
-u_3'
\end{bmatrix}
\right)
}_{\tilde \psi(u_1-u_1', u_2-u_2', -u_3')}
 \d\bfu', 
\end{align*}
where 
$\tilde \psi(\bfu)=u_3\psi(\bfu)$.
 When grouped this way, it shows that the expression is a convolution applied to $\tilde \psi$, evaluated at $(u_1, u_2, 0)$. Since the impulse in $h(\bfu', t)$ from \eqref{eq:h_bfx_t_def} occurs at $t=\norm{\bfu'}$, 
 the factor $t$ can be pulled out as a constant. 
  The expression becomes
\begin{align*}
&
\frac{1}{t}
\left[ h(\bfu, t)\conv{\bfu} \tilde\psi(\bfu) \right]_{u_3=0}.
\end{align*}
This implies
\begin{align}
\bar\scrH\{ \psi\}(u_1,u_2,t)
&=\frac{1}{t} \scrH\{\tilde \psi\}(u_1,u_2,t).
\end{align}

This completes the proof. 
\end{proof}

\subsection{Proof of Proposition 3.4}
\label{sec:proof4}

\begin{proof}

For convenience, in this proof we use $xyz$-coordinates with $z=u_1$, $x=u_2$, $y=u_3$.
This avoids some confusion with the $u_1$, $u_2$ and $t$ variables, which would otherwise switch with the first and second application of $\scrG$.  
Since we are only concerned with applying the impulse response on a 2D slice, 
convolution is limited to $x$ and $y$ variables. 

The relationship between $(u_1,u_2,t)$ and $(x,y,z)$ is somewhat nuanced, as indicated by the following substitutions. 
From Equation~\eqref{eq:G_eq_delta_conv}, we can write
\begin{align}
&\scrG_f\{\psi\}(u_1, u_2, t) 
\nonumber
\\&\quad=
\left[
g_a(x,y,t) \conv{x,y}
\psi(z,x,y)\right]_{x=u_2,y=0,z=u_1}, 
\label{eq:gf_psi_xyz_def}
\end{align}
where
\begin{align}
\label{eq:ga_def_eq}
g_a(x,y, t) = f_a(x, y)\conv{x,y}\tilde g(x,y,t),
\end{align}
and
\begin{align}
\label{eq:g_tilde_def}
\tilde g(x,y, t) = \delta(t - \sqrt{x^2+y^2}).
\end{align}

In order to model a rectangular transducer aperture, we are concerned with a separable cross section of it which is a line segment. 
Let 
$f_a$ 
 be an aperture function 
of width $a$, corresponding to a line segment located at $y=0$ in the 2D $xy$-plane, 
defined by 
\begin{align*}
f_a(x,y)&=\rect\left(\frac{x}{a}\right)\delta(y)
\\&=
\left(\step\left(x+\frac a 2\right) - \step\left(x-\frac a 2\right)\right) \delta(y), 
\end{align*}
where $\step(x)$ is the unit step function.

To simplify the mathematics, first we consider geometry for a semi-infinite ray defined by $$f(x,y)=\step(x)\delta(y).$$ This way, we will be able to form a line segment when needed,  by adding and subtracting two rays shifted by the aperture width $a$ per
\begin{align*}
f_a(x,y)&=f(x+a/2,y)-f(x-a/2,y)
\\
&=(\step(x+a/2)-
\step(x-a/2))\delta(y).
\end{align*}

Therefore, 
the convolution that we are now interested in, which is 
the ray aperture based version of Equation~\eqref{eq:ga_def_eq}, is 
\begin{align*}
&
f(x,y) \conv{x,y} 
\tilde g(x,y,t)
\\
&=\int_{\R^2}f(x-x',y-y')\tilde g(x',y',t)\,\dx'\dy'
\\
&=\int_{\R^2}\step(x-x')\delta(y-y')\tilde g(x',y',t)\,\dx'\dy'.
\end{align*}
The interpretation is that we have to convolve a ray with $\tilde g(x,y,t)$, which is a circularly symmetric object. 
To do this, we will perform the integration for the convolution in polar coordinates.

In polar coordinates, using $x'=r\cos\theta$ and $y'=r\sin\theta$, this becomes
\begin{align*}
\int_{0}^{2\pi}\!\int_{0}^{\infty}\step(x-r\cos\theta)\delta(y-r\sin\theta)\delta(r-t)\,r\dr\d\theta,
\end{align*}
By using the sifting property%
\footnote{
The sifting property is
$\int_{-\infty}^{\infty} \delta(x-y) f(x) \,\d x = f(y)$.
}
in the variable $r$, the inner integral disappears, and the expression becomes
\begin{align*}
t
\int_{0}^{2\pi}\!\step(x-t\cos\theta)\delta(y-t\sin\theta)\,\d\theta.
\end{align*}

Since $\theta$ is not an isolated variable, 
the function $\delta(y-t \sin\theta)$ must be transformed before sifting in the variable $\theta$ can be applied to this integral%
\footnote{
To use the sifting property of the delta function, the coefficient on the variable that is being sifted must be equal to 1 (otherwise using the scaling property of the delta function is necessary). 
This also applies
in more general cases of the form \mbox{$\delta(\bfg(\bfx)-\mathbf{y})$}, where the sifting variable $\bfx$ is found within a function $\bfg(\bfx)$.  
Before sifting can properly be applied,
the inverse $\bfg^{-1}$ must be used 
to isolate $\bfx$,
and the delta function must be scaled. When $\bfg$ is a one-to-one invertible function, this leads to the scaling relationship of the delta function
\begin{align*}
\delta(\bfg(\bfx)-\mathbf{y}))
\equiv
\frac{1}{\abs{\det \bfg'(\bfg^{-1}(\mathbf{y}))}}\delta(\bfx-\bfg^{-1}(\mathbf{y})), 
\end{align*}
where $\bfg'$ is a Jacobian derivative 
(see Equation 2.99 and Section 2.9 of Faris et al.\cite{faris2016shanghai} 
for a rigorous discussion.)}.
To isolate $\theta$, 
we use
the change of variables property\cite{faris2016shanghai} of the delta function
with 
$\theta=\arcsin\left(\frac{y}{t}\right)$.
This
transforms the delta function according to
\begin{align*}
\delta(y-t \sin\theta)\equiv
\frac{1}{\sqrt{t^2-y^2}}
\delta\!\left(\theta-\arcsin\left(\frac{y}{t}\right)\right). 
\end{align*}
Now, sifting in the variable $\theta$ can be applied. 
For a given value of $t$, the term 
$\delta\!\left(\theta-\arcsin\left(\frac{y}{t}\right)\right)$
 is non-zero at two values of $\theta$ on the interval $[0, 2\pi]$. 
This occurs when $\theta=\arcsin(y/t)$, and is valid when $\abs{y}< t$. By trigonometric identity, this means $t\cos\theta=\pm \sqrt{t^2-y^2}$.
 This dual-value 
 causes integration to split into two semi-circular pieces. 
The integral reduces to
\begin{align}
\label{eq:object_solid}
&
\frac{t}{\sqrt{t^2-y^2}}
\Big(\step\!\left(x+\sqrt{t^2-y^2}\right)
\,+
\step\!\left(x-\sqrt{t^2-y^2}\right)\Big)
\nonumber
\\
&\qquad \times
\Big(\!\step(y+t)-\step(y-t)\Big)
.
\end{align}
Here, the expression $\step(x+\sqrt{t^2-y^2})$ is a step function that is shifted by an amount equal to the left half of a circle of radius $t$. The expression $\step(x-\sqrt{t^2-y^2})$ is a step function shifted 
for the right
 half of the circle. Since complex values are not allowed, then $\abs{y}$ must be less than or equal to $t$. This is accomplished by including a factor of $(\step(y+t)-\step(y-t))$.

Recall that equation \eqref{eq:object_solid} was developed using the ray aperture $f(x,y)$ to simplify the mathematics. Now, we can consider the line segment aperture $f_a(x,y)$. This means adding and subtracting equation \eqref{eq:object_solid}
when it is shifted by $-\frac a 2$ and $\frac a 2$. The resulting form of
 Equation~\eqref{eq:ga_def_eq} is written 
\begin{align}
\label{eq:ga_full}
&g_a(x,y, t)= \nonumber 
\\&
\nonumber
\frac{t}{\sqrt{t^2-y^2}}
\Big(\!\step(y+t)-\step(y-t)\Big)
\ifdefined\SHORTTEX
\eqbreak\times
\else
\fi
\Bigg[\rect\!\left(\frac{x+\sqrt{t^2-y^2}}{a}\right)
+
\rect\!\left(\frac{x-\sqrt{t^2-y^2}}{a}\right)\Bigg]
. 
\end{align}

Equation \eqref{eq:ga_full} corresponds to a solid object (not an outlined object).
This can be seen because the $\rect$ function has a width equal to $a$.  
Therefore, due to the width of the transducer, the 2D spatial convolution $g_a(x,y,t) \conv{x,y} \psi(x,y)$ is not sparse, 
which leads to some inefficiency%
.
 
 Convolution with an object that is sparse  has only a few non-zero points that 
must be sampled and multiplied.
The dense spatial convolutions can be avoided if there is a time-domain filtering operation that \enquote{fills in} points that have been spatially sampled sparsely.  
As it turns out, the time derivative of an approximation of equation \eqref{eq:ga_full} is sparse spatially, and integration in the time-domain can be used to compute the convolution efficiently.

A far-field approximation can simplify equation \eqref{eq:ga_full} by assuming 
small transducer aperture width $a$ relative to working distances. 
To arrive at this approximation, equation \eqref{eq:ga_full} is written in expanded form as
\begin{align*}
&t 
\int_{x'=-\infty}^{\infty}\!
\int_{y'=-t}^{t}\!
\frac{\psi(x-x', y-y')}{\sqrt{t^2-y'^2}}
\\
&\quad\times
\left[
\!\int_{-\frac a 2}^{\frac a 2}\delta\left(x'-x''+\sqrt{t^2-y'^2}\right)\,\d{x''}\right.
\\
&\quad
\left.
+\!\int_{-\frac a 2}^{\frac a 2}\delta\left(x'-x''-\sqrt{t^2-y'^2}\right)\,\d{x''}
\right]
\,\d{x'}\d{y'}.
\end{align*}

By examining the inner integrals above, over the region of integration  $\abs{x''} \ll \frac{a}{2}$.  
From the delta functions, 
 the condition for non-zero that must be met is $(x'-x'')\pm\sqrt{t^2-y'^2}=0$. 
Also, if it is assumed that $\abs{x''} < a \ll \abs{x'}$ (i.e. if ${x''}$ is small), then 
we can use the approximation $x' - x'' \approx x'$. 
This implies that we can also use the approximation $\sqrt{t^2-y'^2} \approx \abs{x'}$ for the term in the denominator, when $x'$ is not too small (i.e. if $x' \ge a $). Otherwise (if $x' \le a$), since the integration with $x''$ is performed over the interval $[-\frac a 2, \frac a 2]$, which has a width of $a$, then the denominator can be replaced with $a$. 
In either case, the denominator no longer has the dependence on $t$ or $y'$, which leads to simplification in the next step. 

Therefore, an approximation for $g_a$ is 
\begin{subequations}
\begin{align}
g_a(x,y, t) \approx
\frac{t}{\gamma(x)}
\rect\!\left(\frac{y}{2t}\right)
\rect\!\left(\frac{x\pm\sqrt{t^2-y^2}}{a}\right),
\end{align}
where
\begin{align}
\gamma(x) = \begin{cases}
\abs{x}&\text{if $\abs{x} \ge a$},\\
a &\text{otherwise}.
\end{cases}
\end{align}
\end{subequations}
 When used in spatial convolution, the factor $t$ remains constant, and it can be factored out 
and processed independently from the spatial convolution. 
To compute $g_a$, a step of integration and differentiation, which cancel each other out, can be placed between multiplication and division by $t$. This yields the identity
\begin{align}
\label{eq:ga_t_int_diff_t}
g_a(x,y,t)=
t \int_{0}^t \frac{\partial}{\partial t} \left( \frac{1}{t}g_a(x,y,t)\right) \,\dt.
\end{align}

Now we return to considering a ray aperture rather than a line segment. 
For a ray aperture, (i.e., by applying the approximation to Equation \eqref{eq:object_solid}), 
the approximation is
\begin{align}
\label{eq:object_solid_ray}
&
\frac{t}{\gamma(x)}
\Big(\step\!\left(x+\sqrt{t^2-y^2}\right)
\,+
\step\!\left(x-\sqrt{t^2-y^2}\right)\Big)
\nonumber
\ifdefined\SHORTTEX
\eqbreak \times
\else
\fi
\Big(\!\step(y+t)-\step(y-t)\Big)
.
\end{align}
Observe that here the factor $\frac{1}{\gamma(x)}$ does not depend on $t$. 
We divide \eqref{eq:object_solid_ray} by $t$, compute its time-derivative, and label the result as $\varsigma(x,y,t)$. This is a sparse object, shown in Figure~\ref{fig.arcpath}a, which is equal to
\begin{align}
&\varsigma(x,y,t)=
\nonumber
\\
&\frac{t}{\sqrt{t^2-y^2}}\left(
\delta\left(x+\sqrt{t^2-y^2}\right)
-\delta\left(x-\sqrt{t^2-y^2}\right)\right)
\nonumber
\\&\qquad \times
\frac{1}{\gamma(x)}
\Big(\!\step(y+t)-\step(y-t)\Big)
\nonumber
&\quad
\\
&\qquad +
\frac{1}{\gamma(x)}
\Big(\!\delta(y-t)+\delta(y-t)\Big)\step(x)
.
\end{align}
For the line segment aperture, this can be used to compute
$\frac{\partial}{\partial t}\left(\frac 1 t \scrG_a\right)$ from $\varsigma(x,y,t)$ as
\begin{align*}
&\frac{\partial}{\partial t}
\left(\frac 1 t \scrG_a\{\psi\}(u_1, u_2, t) \right)
\\
&\qquad=
\left[
h_a(x,y,t)
\conv{x,y}\psi(u_1,x,y)\right]_{x=u_2,y=0}
\end{align*}
where
\begin{align*}
h_a(x,y,t) = \varsigma\left(x+\frac a 2,y,t\right) - 
	\varsigma\left(x-\frac a 2,y,t\right). 
\end{align*}
Writing out the convolution integral, it becomes
\begin{align*}
\int_{\R^2}
h_a(x',y',t)\,\psi(u_1, u_2-x', 0-y')
\dx'\dy'.
\end{align*}
This is integrated over $\dx'\dy'$. However, $h_a$ consists of 6 contours: two positive circular arcs, two negative circular arcs, and two lines. The circular arcs are of the form
\begin{align*}
\frac{1}{\gamma(x\pm a)}
\frac{t}{\sqrt{t^2-y^2}}
\delta\left(x\pm a \pm\sqrt{t^2-y^2}\right)\step(x\pm a)
.  
\end{align*}
This can be integrated in polar coordinates after a change of coordinates. 
Thus
\begin{align*}
&\int_{\R^2}
\frac{t}{\sqrt{t^2-y'^2}}
\delta\left(x' -\sqrt{t^2-y'^2}\right)\step(x')\frac{1}{\gamma(x')}
\ifdefined\SHORTTEX
\eqbreak\qquad\qquad\qquad
\else
\fi
\psi(u_1, u_2-a-x', -y')   
\dx'\dy'
\\&=
\int_0^\infty \!\!\int_k^{\pi+k}\!
\delta\left(t - r\right)
\frac{\psi(u_1, u_2-a-\cos\theta, \sin\theta) }  {\tilde\gamma(r, \theta)}
r dr\d\theta
\\&=
 \!\!\int_k^{\pi+k}\!
 \frac{\psi(u_1, u_2-a-\cos\theta, \sin\theta)   }
{\tilde\gamma(t, \theta)}
 \,
\d\theta.
\end{align*}
In this equation, we use $\tilde\gamma(t, \theta) = \gamma(\abs{t \cos \theta})$. 
In this manner, the four arcs of $\frac{\partial}{\partial t}\frac 1 t\scrG_a$ can be computed 
as listed in Proposition~\ref{prop:s4}.  
Similarly the two lines can also be computed in this manner. 
This completes the proof. 
\end{proof}

\ifdefined\SHORTTEX
\else
\egroup
\fi
\end{subappendices}


\section*{Acknowledgement}
\addcontentsline{toc}{section}{Acknowledgement}

We acknowledge the support of the Natural Sciences and Engineering Research Council of Canada (NSERC), funding reference number RGPIN-2017-06496. 
Portions of this work were performed in collaboration with Seno Medical Instruments, Inc. 
The authors would like to thank M. J. Moore, Y. Xu, C. J. Kumaradas, B. A. Clingman, S. C. Miller, T. G. Miller and D. G. Herzog for feedback or discussion regarding this manuscript.



\printbibliography


\end{document}